\documentclass[12pt]{article}
\usepackage[affil-it]{authblk}
\usepackage{multirow}
\usepackage{titling}
\usepackage{subcaption}
\usepackage{algorithm}
\usepackage[noend]{algpseudocode}
\usepackage{chngcntr}
\usepackage{apptools}
\usepackage{comment}
\usepackage[multiple]{footmisc}

\usepackage{epigraph}
\usepackage{graphicx}
\usepackage{color}
\usepackage{enumitem}
\setlist[enumerate]{leftmargin=1.5cm,rightmargin=0.5cm,noitemsep, topsep=2pt}

\usepackage{etoolbox}
\newcommand{\zerodisplayskips}{%
  \setlength{\abovedisplayskip}{4pt}%
  \setlength{\belowdisplayskip}{4pt}%
  \setlength{\abovedisplayshortskip}{4pt}%
  \setlength{\belowdisplayshortskip}{4pt}}
\appto{\normalsize}{\zerodisplayskips}
\appto{\small}{\zerodisplayskips}
\appto{\footnotesize}{\zerodisplayskips}

\usepackage{algorithm}
\usepackage[noend]{algpseudocode}

\usepackage{mdframed}

\definecolor{clemson-orange}{RGB}{234,106,32}
\definecolor{chicago-maroon}{RGB}{128,0,0}
\definecolor{northwestern-purple}{RGB}{82,0,99}
\definecolor{cornell-red}{RGB}{179,27,27}
\definecolor{sauder-green}{RGB}{171,180,0}
\definecolor{gray}{RGB}{192,192,192}
\definecolor{lawngreen}{RGB}{0,250,154}
\definecolor{pink}{RGB}{255,0,128}

\usepackage[round]{natbib}  
\makeatletter
\def\BState{\State\hskip-\ALG@thistlm}
\makeatother

\usepackage{amsmath}
\usepackage{amssymb}
\usepackage{amsthm}
\allowdisplaybreaks

\newcommand{\bb}{\mathbb}

\newcommand{\R}{\bb R}

\DeclareMathOperator\supp{supp}

\DeclareMathOperator\dc{dc}

\DeclareMathOperator{\argmax}{arg\,max}

\theoremstyle{definition}
\newtheorem{theorem}{Theorem}
\newtheorem{lemma}{Lemma}
\newtheorem{corollary}{Corollary}
\newtheorem{proposition}{Proposition}

\newtheorem{definition}{Definition}

\newtheorem{example}{Example}

\newtheorem{step}{Step}
\AtAppendix{\counterwithin{lemma}{section}}
\AtAppendix{\counterwithin{proposition}{section}}
\AtAppendix{\counterwithin{corollary}{section}}

\theoremstyle{definition}
\newtheorem{claim}{Claim}

\usepackage[centering,marginparwidth=0.9cm]{geometry}
\usepackage{etoolbox}
\usepackage{lipsum}
\makeatletter
\patchcmd{\@addmarginpar}{\ifodd\c@page}{\ifodd\c@page\@tempcnta\m@ne}{}{}
\makeatother
\newcommand{\mar}[1]{}

\usepackage{fullpage}
\usepackage[colorlinks,citecolor=northwestern-purple,urlcolor=chicago-maroon,linkcolor=chicago-maroon,backref=page,hyperfootnotes=false]{hyperref}

\usepackage[nameinlink]{cleveref}
\crefname{assumption}{Assumption}{Assumptions}
\crefname{lemma}{Lemma}{Lemmas}
\crefname{theorem}{Theorem}{Theorems}
\crefname{corollary}{Corollary}{Corollaries}
\crefname{proposition}{Proposition}{Propositions}
\crefname{claim}{Claim}{Claims}
\crefname{procedure}{Procedure}{Procedures}
\crefname{algorithm}{Algorithm}{Algorithms}
\crefname{figure}{Figure}{Figures}
\crefname{remark}{Remark}{Remarks}
\crefname{section}{Section}{Sections}
\crefname{procedure}{Procedure}{Procedures}
\crefname{example}{Example}{Examples}
\crefname{definition}{Definition}{Definitions}
\crefname{table}{Table}{Tables}
\crefname{equation}{}{}
\crefname{enumi}{}{}
\crefname{conjecture}{Conjecture}{Conjectures}
\crefname{step}{Step}{Steps}
\crefname{appendix}{Appendix}{Appendices}
\crefname{footnote}{Footnote}{Footnotes}

\title{``Near'' weighted utilitarian characterizations of Pareto optima\thanks{We thank the Co-editor, Asher Wolinsky, as well as four anonymous referees for comments that led to major improvements. We are also grateful to Florian Brandl, Timothy Y Chan, Alexander Engau, Atsushi Kajii, Yuichiro Kamada, Michihiro Kandori, and Eitetsu Ken for their helpful comments and conversations, as well as to seminar audiences at Arizona State, Carlo Alberto, Carlos III, Carnegie Mellon, Harvard-MIT, KAIST, Northwestern, Stanford, University of British Columbia, UC-Davis,  and the International Conference on Game Theory at Stony Brook.  We are especially grateful to Ludvig Sinander and Gregorio Curello for a question that led to this project. We acknowledge research assistance from Yutaro Akita, Nanami Aoi, Xuandong Chen, William Grimme, Jiangze Han, Yusuke Iwase, Masanori Kobayashi, Kevin Li, Leo Nonaka, Ryo Shirakawa, Shoya Tsuruta, Ayano Yago, and Yutong Zhang.  Yeon-Koo Che is supported by National Science Foundation Grant SES-1851821. Fuhito Kojima is supported by the JSPS KAKENHI Grant-In-Aid 21H04979. Christopher Thomas Ryan is supported by the Natural Sciences and Engineering Research Council of Canada Discovery Grant RGPIN-2020-06488 and the UBC Sauder Exploratory Grants Program.  This work was supported by the Ministry of Education of the Republic of Korea and the National Research Foundation of Korea  (NRF-2020S1A5A2A03043516).
}}

\author{
Yeon-Koo Che%
\thanks{Department of Economics, Columbia University} \qquad
Jinwoo Kim%
\thanks{Department of Economics and SIER, Seoul National University} \qquad
Fuhito Kojima%
\thanks{Department of Economics and Market Design Center, University of Tokyo}  \\
Christopher Thomas Ryan%
\thanks{Operations and Logistics Division, UBC Sauder School of Business, University of British Columbia} 
}
\affil{}

\thanksmarkseries{arabic}

\begin{document}

\maketitle

\begin{abstract}	
 We characterize Pareto optimality via  ``near'' weighted utilitarian welfare maximization.  One characterization sequentially maximizes utilitarian welfare functions using a finite sequence of nonnegative  and eventually positive welfare weights. The other maximizes a  utilitarian welfare function with a certain class of positive hyperreal weights. The social welfare ordering represented by these ``near'' weighted utilitarian welfare criteria are characterized by  the standard axioms for weighted utilitarianism under a suitable weakening of the continuity axiom.
\\


\end{abstract}

\section{Introduction}\label{s:introduction}

Pareto optimality  is a central concept in economics for its normative appeal. 
Also central is weighted utilitarian welfare maximization;  e.g., \citet{harsanyi1955cardinal} famously defended it as a social welfare function based on several normative axioms.
Moreover, weighted utilitarianism is widely invoked in practice, including applied research and policy debates.  Given the prominent roles played by these two concepts, attempts have been made to establish a connection between Pareto optima and weighted utilitarianism---or more precisely, a characterization of Pareto optima via  weighted utilitarian welfare maximization. Yet, such a characterization has so far been elusive.

It is well known  that, given a closed and convex utility possibility set, which we assume throughout,  every Pareto optimal utility vector maximizes some \emph{nonnegatively} weighted sum of utilities of agents (see Proposition 3.45 in \cite{bewley2009general}). But the converse is false:  not every such maximizer is Pareto optimal. To see this, suppose a society consists of two agents, 1 and 2, and the utility possibility set is given by $U$ in \Cref{fig:illustration-of-sequential-converse}.
\begin{figure}[htb]
	\centering
	\includegraphics[scale=.9]{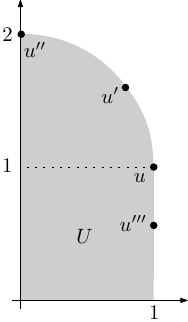}
	\caption{Weighted utilitarian welfare maximization need not yield a Pareto optimum.}
		\label{fig:illustration-of-sequential-converse}
\end{figure}
All points on the ``outer'' boundary, including the vertical segment,  maximize suitably weighted sums of agents' utilities within $U$, but not all of them are Pareto optimal. In particular, the points on the vertical segment strictly below $u$, such as $u'''$, all maximize the utility sum with weights $\phi=(1,0)$---i.e., only 1's utility.  Yet, none of these points is Pareto optimal.  The reason is that the welfare of the agent receiving zero weight is \emph{not} counted. 
 
By contrast, if weights are restricted to be (strictly)  positive for all agents, weighted utilitarian welfare maximization does always yield a Pareto optimum (Proposition 3.23 of \cite{bewley2009general}).  But  the converse is false: not every Pareto optimal outcome can be obtained in this way.  In \Cref{fig:illustration-of-sequential-converse}, $u'$ is Pareto optimal and obtained by weighted utilitarian welfare maximization with positive weights, but  $u$ and $u''$, which are also Pareto optimal, cannot be obtained.

While positive welfare weights do not yield points like $u$ in  \Cref{fig:illustration-of-sequential-converse}, one may conjecture that they may ``in the limit''; for instance, $u$ is a limit of welfare-maximizing utility vectors with positive weights $(1, 1/n)$, as $n\to \infty.$ Indeed, \cite{arrow1953abb} show that  every Pareto optimal vector is a limit of a sequence of utility vectors that maximize some positively weighted sum of utilities---a result known as  the  ABB theorem.%
%
%
\footnote{This theorem has spawned a series of extensions  to spaces more general than Euclidean space.  See \cite{daniilidis2000arrow} for a survey of ABB theorems.}
%
%
Unfortunately, this too does not lead to a characterization when there are more than two agents:%
%
%
%
\footnote{When there are two agents, the limit $u\in U$ of any sequence $\{u^k\}$ of utilities $u^k\in U$ maximizing a positively weighted sum of utilities is Pareto optimal, where $U$ is the utility possibility set, assumed to be closed and convex.  To see it, let $\{\phi^k\}$ be the sequence of positive weights, normalized to be in the simplex, such that $u^k \in \arg\max_{(u'_1,u'_2)\in U} \sum_{i=1}^2\phi^k_i u'_i$,  and let $\phi$ denote its limit (say of a convergent subsequence).  Clearly, $u \in   \arg\max_{(u'_1,u'_2)\in U} \sum_{i=1}^2\phi_i u'_i$.  If $\phi_1$ and $\phi_2$ are both positive, then $u$ is Pareto optimal, so assume  $\phi_1=1$ and $\phi_2=0$ without loss.  Suppose for contradiction $u$  is not Pareto optimal.  Then, there must exist $v  \in U$ such that $v_1=u_1$ and $v_2 >u_2$, where the equality holds since $u \in    \arg\max_{(u'_1,u'_2)\in U} \sum_{i=1}^2\phi_i u'_i  = \arg\max_{(u'_1,u'_2)\in U} u'_1$.  Since $u^k$'s are all Pareto optimal, we have $u_1^k \le v_1=u_1 $ and $u^k_2\ge v_2 >u_2$ for all $k$, so $u^k$ never converges to $u$, a contradiction.}
%
%
again its converse is false---namely, a limit point of such a sequence may not be Pareto optimal. To see this, suppose there are three agents, 1, 2, and 3,  with possible utility vectors depicted in  \cref{fig:tilted-cone}.  
   \begin{figure}
   	\centering
   	\includegraphics[scale=.7]{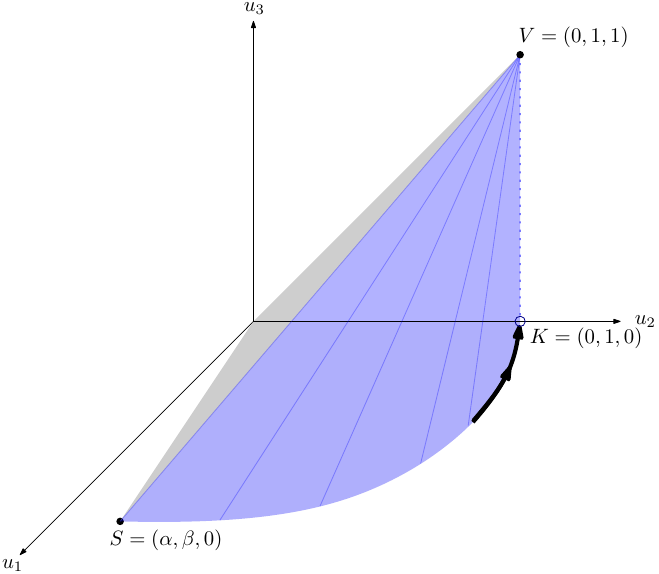}
   	\caption{ The ``tilted cone'' adapted from \cite{arrow1953abb} and \cite{bitran1979structure}. The set is the convex hull of the portion of the unit disk centered at the origin in the $u_1$-$u_2$ plane from point $K$ to point $S$ (where $\alpha^2 + \beta^2 = 1$ with $\alpha \in (0,1)$) and the apex point $V = (0,1,1)$.  The blue surface, including all of its boundaries except for the dotted line, is the set of Pareto optimal utility vectors.}
   	\label{fig:tilted-cone}
   \end{figure}
   The point $K$ is a limit of the sequence of points  maximizing a positively weighted sum of utilities (see the arrow) but is Pareto dominated, say, by the point $V$.
   \begin{figure}
    \centering
    \includegraphics[scale=.4]{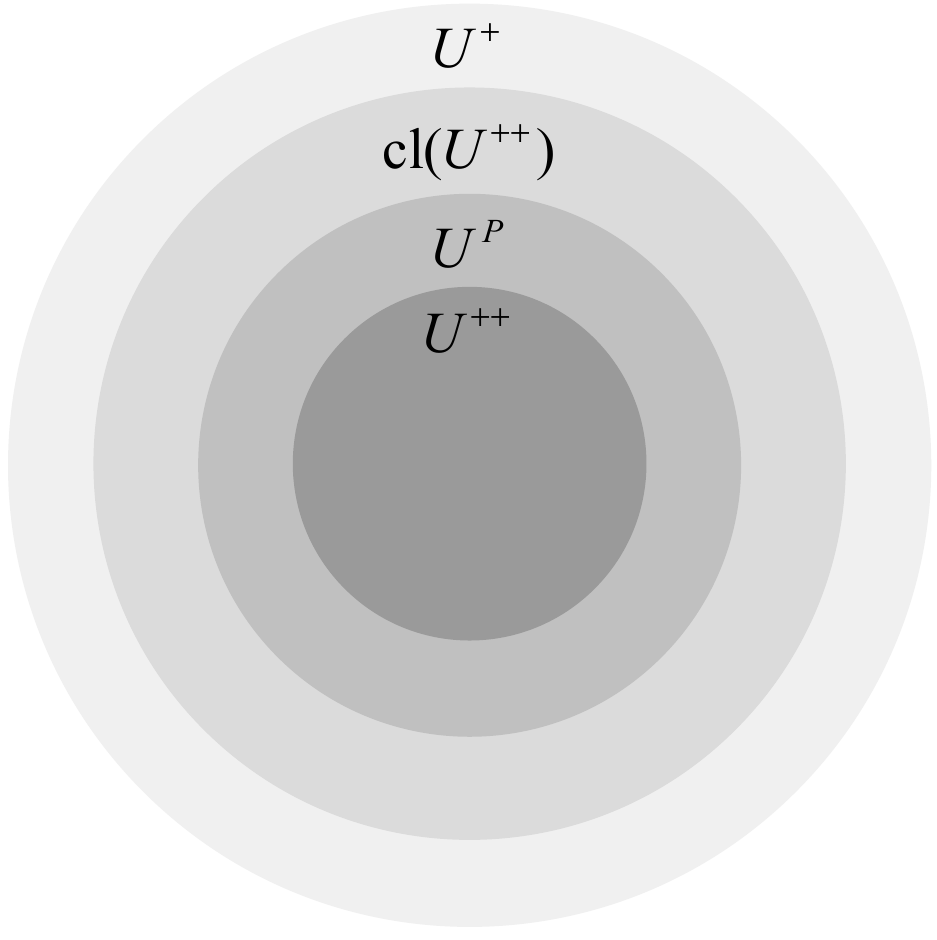}
    \caption{Alternative notions of utilitarian welfare maximization in relationship with Pareto optimality. 
    The containment $U^{++} \subset U^P \subset U^+$ follows from Propositions 3.23 and 3.45 in \cite{bewley2009general}. The containment $U^P \subset \text{cl}(U^{++})$ is from \cite{arrow1953abb}. The containment $\text{cl}(U^{++}) \subset U^+$ is straightforward. }
    \label{fig:venn-diagram}
  \end{figure}  
The relationship between Pareto optima and the alternative notions of weighted utilitarianism is depicted in   \cref{fig:venn-diagram}, where  $U^P$ is the set of  Pareto optimal utility vectors  while  $U^+$ and $U^{++}$ are the sets of utility vectors that maximize nonnegatively weighted and (strict) positively weighted utilitarian welfare, respectively, with $\text{cl}(U^{++})
$ being the closure of $U^{++}$.

 This paper provides exact characterizations of Pareto optima by close variants of weighted utilitarian welfare maximization. To ease language, we will refer to weighted utilitarian welfare maximization simply as \emph{utilitarianism}.\footnote{In particular, note that we drop the qualifier ``weighted'' in our usage of the concept of utilitarianism while keeping in mind that utilitarianism is always used in the weighted sense. Indeed, we have no occasion to discuss unweighted utilitarianism. It is only for emphasis or to provide further clarification that  we use the qualifier ``weighted'' in connection to  utilitarianism.}
 
 We show that a utility vector $u$ is Pareto optimal if and only if there exists a finite sequence of nonnegative and ``eventually positive'' welfare weights such that in each round $t$, $u$  maximizes the round-$t$ weighted sum of utilities out of those surviving from round $t-1$.   Here, ``eventually positive'' means that the support of the weight vector strictly grows over the rounds with the weight vector in the final round having full support.
 
To illustrate why our SUWM successfully characterizes Pareto optima, let us   revisit why neither the ``nonnegative utilitarianism'' captured by $U^{+}$  nor the ``positive utilitarianism'' captured by $U^{++}$ in \Cref{fig:venn-diagram}  works.   Nonnegative utilitarianism can include Pareto suboptimal outcomes because some individual's utility may not ``count'' at all.  
Positive utilitarianism avoids this problem  by requiring that every individual's utility carry  positive weights.  However,  as can be seen from \Cref{fig:illustration-of-sequential-converse}, it excludes   Pareto optimal outcomes that can be achieved only by assigning some individuals ``infinitely smaller'' weights than others. SUWM resolves this seeming conflict by assigning positive weights to individuals so that ``every agent's welfare counts'' but in different rounds:    Individuals with strictly positive weights only in later rounds can be regarded as carrying infinitely smaller weights than those with   positive weights in earlier rounds. 

The preceding observation gives rise to our second characterization of Pareto optimality, via one-shot maximization of utilitarian welfare with \emph{hyperreal} weights.  Hyperreal numbers  include not only standard real numbers but also infinitesimal ``numbers.'' The space of hyperreal numbers is very large, which may limit the usefulness of the characterization. By contrast, our characterization places an added discipline and structure on such social welfare functions.  The resulting criterion, called {\it simple hyperreal utilitarian welfare maximization} (SHUWM), requires the hyperreal weights to be not only strictly positive but also represented by a finite sequence of nonnegative and eventually positive real weights.   

Both of these characterizations of Pareto optimality capture the essential feature of standard weighted utilitarian welfare maximization. First, SUWM and SHUWM reduce to utilitarianism in many situations in which the former involves one-round maximization and the latter involves no infinitesimal weights.  Second, the welfare functions used in these characterizations are inherently \emph{linear} (based on weighted sums of agent utilities), albeit with SUWM having several rounds of linear optimizations and SHUWM involving hyperreal weights. Third, a consequence of this linearity is that the utilities of individuals are aggregated by weights that do not depend on the particular utility profile under consideration, a property we refer to as having ``constant weights.'' This is in contrast to other social welfare functions, such as Rawlsian and leximin whose weighting of an agent's utility depends on her relative position in  a given utility vector. 

The sense in which our characterizations constitute ``near'' utilitarianism is further clarified by the social welfare orderings that  underpin  our characterizations.  
\cite{d2002social} show that for social welfare orderings to be represented by a utilitarian welfare function, they must not only satisfy the \textsf{Pareto Principle}---namely, they must preserve Pareto domination order---but they must also satisfy two additional axioms: 
 \textsf{Invariance} and \textsf{Continuity}.  \textsf{Invariance} requires the orderings to be robust to translation and/or scaling of the utility profiles of individuals.  \textsf{Continuity} requires the orderings to be robust to perturbations of utility profiles. \textsf{Continuity} effectively forces the welfare weights of  agents to be in the same order of magnitude, thus making it impossible for the weight of an agent to be infinitesimally smaller than that of another agent.  Since the latter feature is crucial for characterizing Pareto optima, \textsf{Continuity} must be relaxed.
 
 Indeed, we show that the welfare orderings associated with SUWM and SHUWM can be obtained by the same set of axioms under a suitable weakening of \textsf{Continuity}---more precisely, by the \textsf{Pareto Principle}, \textsf{Invariance}, and \textsf{Weak Continuity}.  The last axiom weakens \textsf{Continuity} by requiring welfare orderings to be robust to perturbations of utilities  of {\it some, but not necessarily all}, individuals,  which is in line with our characterization of Pareto optima that allows  some individuals to be assigned infinitely larger weights than others. That our welfare notions preserve a version of continuity, albeit weakened, is a nontrivial marker of the sense in which SUWM and SHUWM closely resemble utilitarianism.  In particular, the same marker is not shared with other reasonable characterizations. For instance, as we show in a subsequent section, an (unrestricted) hyperreal-weighted utilitarian welfare function does not satisfy \textsf{Weak Continuity}.

Our characterizations of Pareto optimality fulfill a long-standing intellectual pursuit of providing a weighted utilitarian foundation for Pareto optimality. In addition, our characterizations of Pareto optimality serve other useful purposes. 

First, the SUWM characterization could provide a tractable method for computing Pareto optimal allocations, which may be useful in the market design context. In fact, SUWM can be viewed as a generalization of the serial dictatorship mechanism in which each agent acts sequentially according to serial order to maximize her utility.   Serial dictatorship is used widely for Pareto optimally allocating indivisible resources when monetary transfers cannot be used.  For instance, serial dictatorship with a randomized serial order---known as random serial dictatorship---is used for assigning public school seats, public and campus housing, and human organs.  One could imagine that SUWM can serve a  similar practical purpose, but in a much more general setting that goes beyond a one-to-one assignment. In each round, one can let a group of agents negotiate over feasible allocations at that round, as will be made precise in \Cref{s:discussion}.  Indeed, a procedure like this is used in the assignment of campus housing.\footnote{\label{fn:cohort-SD} For example, the campus housing assignment at Columbia university uses a cohort-based serial dictatorship, in which a group of students chooses a suite collectively in each round of the serial dictatorship procedure; presumably, the students then negotiate among themselves to allocate rooms within the assigned suite.}  Alternatively, a central clearinghouse may compute an optimal choice for the group in each round.\footnote{In both scenarios, we are implicitly assuming complete information.  In case agents' preferences are unobserved, the designer must rely on their preference reports, in which case agents' incentives become an important aspect of the market design.  While this issue is beyond the scope of the current paper, it can be addressed in some specific settings such as cohort-based serial dictatorship mentioned in \cref{fn:cohort-SD}, where the standard strategy-proofness property would extend to a group of students as long as they know their preferences.}

Second, the SUWM characterization could serve as a useful analytical tool for analyzing the behavior of Pareto optima as a set.  For instance, one may study the comparative statics of Pareto optima---i.e., how they change as the primitives change---utilizing monotone comparative statics methods developed for optimization (e.g., \cite{topkis:98} and \cite{milgrom/shannon:94}). The ``round-wise'' linear structure of SUWM admits a convenient aggregation property that is crucial for such an analysis. Indeed, \cite{CKK2019} use this property to develop a theory of monotone comparative statics of Pareto optima:  they show that  when agents' utility functions shift in a way that leads to higher {\it individual} choices of decisions (e.g., \cite{milgrom/shannon:94}), the Pareto optima shift to a higher set of actions in a suitable sense.\footnote{In particular,  properties such as supermodularity and increasing differences, which are important for the monotone comparative statics analysis, are preserved under this  aggregation. The same proof would not have been possible with nonlinear welfare functions.}

The remainder of the paper is organized as follows. \cref{s:main-result} describes our setting and establishes a few preliminaries used in our main results. \cref{pareto-optimality-section} states our characterization of Pareto optimality. Here, we discuss the tools used to prove the result.  \cref{social-preference-section} establishes the axiomatization of SUWM and SHUWM. \cref{s:discussion} looks at other reasonable characterizations of Pareto optimality that fail at least one of the ``near'' utilitarian axioms set out in the previous section. \cref{s:conclusion} concludes with some suggestions for future work. The  appendix provides proofs of our main characterization and axiomatization results. A  supplementary appendix contains statements and proofs of additional results.

\section{Setting and preliminaries}\label{s:main-result}

In this section, we introduce our basic setting and introduce some elementary concepts needed for stating our main results. 

Let $I=\{1,2,\ldots,n\}$ denote a finite set of agents and the \emph{utility possibility set} $U \subset \R^n $
be the set of possible utility vectors the agents may attain. We assume that $U$ is closed and convex. If $U$  stems from  an underlying choice space $X$ via utility functions $(u_i)_{i \in I}: X \to \R^n$, then we let  
\begin{equation} \label{eq:X-U}
	U =  \{   u \in \R^n \mid  u \le u(x) \text{ for some } x \in X \}.%
 \footnote{To be precise,  the utility possibility set is  often defined as $\{  u \in \R^n \mid   u  = (u_i (x))_{i\in I} \mbox{ for some } x \in X  \}$, which differs from \cref{eq:X-U}.  However, the two sets share the same set of Pareto optima since those points are on the common outer boundary of the sets. Thus, formulating the set $U$ either way makes no difference for our results while the current formation facilitates our analysis.}
\end{equation}  

That $U$ is closed and convex is arguably a mild assumption that is satisfied if, for instance, $U$ is induced by utility functions $(u_i)_{i\in I}$  that are upper semicontinuous and concave on a choice set $X$ that is compact and convex.%
%
%
\footnote{Note that compactness and convexity of the choice set $X$ are satisfied if, for instance, all lotteries of social outcomes, which are in turn finite, or more generally compact, are feasible.}
%
%

For any $u,v\in \mathbb R^n$, we write $v \ge u$ if $v_i \ge u_i$ for all $i \in I$, $v > u$ if $v \ge u$ and $v \ne u$, and  $v\gg u$ if $v_i > u_i$ for all $i \in I$.    We say a point  $u$ in $U$ is \emph{Pareto optimal} with respect to $U$ if there exists no $v \in U$  with $v > u$.  Let $U^P\subset U$ denote the set of all Pareto optimal points (or, more simply, Pareto optima). 

For any $\phi \in \R^n$, consider the optimization problem: 
\begin{align}\label{eq:linear-optimization}
	\max_{u \in U} \langle \phi, u \rangle,
\end{align}
where $\langle \phi, u \rangle := \sum_{i=1}^n \phi_i u_i$. We call $\phi$ a \emph{weight vector}.
Throughout the paper, we only consider nonzero weight vectors (i.e., $\phi \ne 0$). We say a  point $u \in U$ \emph{maximizes} the weight vector $\phi$ over $U$ (or simply \emph{maximizes} $\phi$) if $u$ is a solution  to \cref{eq:linear-optimization}. We call a weight vector $\phi$  \emph{nonnegative} if  $\phi >0$  and \emph{positive} if $\phi \gg 0$. For any vector $v \in \R^n$, the \emph{support} of $v$ is the set of indices where $v$ is nonzero; i.e., $\supp v:= \left\{ i \in I  \mid v_i \neq 0\right\}$. A positive $\phi$ has full support; i.e., $\supp \phi = I$.

Our discussion uses the language of hyperreal numbers. We introduce the basics here. The set of \emph{hyperreal numbers} $^* \mathbb R$ consists of real numbers as well as ``infinite'' and ``infinitesimal'' numbers. Infinite numbers are larger than any real number. Infinitesimal numbers (or simply infinitesimals) are closer to $0$ than any real number. A formal definition of $^* \mathbb R$ is somewhat tedious and we will not reproduce it here. Instead, we refer the reader to \cite{goldblatt2012lectures}. Although the use of hyperreal numbers (in what is termed \emph{nonstandard analysis}) is not completely standard in economics, it has been used in a variety of settings including choice under uncertainty  \citep{blume1991lexicographic}, game theory \citep{dilme2022lexicographic}, and exchange economies   \citep{brown1975nonstandard}. See \cite{anderson1991non} for a survey of applications of nonstandard analysis to economics.

Important properties of the set of hyperreal numbers for our purposes are that (i) $^* \mathbb R$ contains a (positive) infinitesimal number, i.e., an element $\epsilon \in {^*\mathbb R}$ such that $\epsilon < r$ for every positive real number $r$ while $\epsilon>0$, and that (ii) arithmetic operations such as addition and multiplication, as well as order relations, are well defined and extended from $\mathbb R$ to $^*\mathbb R$ in expected ways.

\section{Characterizations of Pareto optimality}\label{pareto-optimality-section}

This section presents our main result, \cref{theorem:characterize-PO}, that provides two alternative ``near'' (weighted) utilitarian characterizations of the set $U^P$ of Pareto optimal points of a given closed convex set $U$. To state these characterizations, we first introduce some additional terminology and definitions. These definitions will be interpreted  after the statement of \cref{theorem:characterize-PO}.

A sequence $\Phi = (\phi^1, \phi^2, \dots, \phi^T)$ of  weight vectors is \emph{nonnegative} if $\phi^t$ is nonnegative for every $t \in \{1,\dots,T\}$. We say that a  sequence $\Phi$ of weight vectors is \emph{eventually positive} if   $ \supp \phi^{t-1}  \subsetneq \supp \phi^t$ for all $t =2, \ldots, T$ and  $ \supp \phi^T =I$. Note that eventual positivity implies   $T \le n$ since the support \emph{strictly} grows along the sequence.

\begin{definition}[Sequential utilitarian welfare maximization  (SUWM)]\label{def:optimizing-set-of-normals}
We say $u \in U$ \emph{sequentially maximizes} a sequence $\Phi = (\phi^1, \phi^2, \dots, \phi^T)$ of  weight vectors over $U$ if \begin{equation} 
u \in  U^t :=  \arg\max_{u' \in  U^{t-1}}  \langle  \phi^t, u'  \rangle, \mbox{ for each }  t = 1,\dots,T,   \label{eq:define-u-t}
\end{equation}
where $U^0=U$. We say $u \in U$ \emph{sequentially maximizes utilitarian welfare} over $U$---or, more simply, $u$ is an \emph{SUWM solution} of $U$---if there exists a sequence $\Phi$ of nonnegative and eventually positive weight vectors such that $u$ sequentially maximizes $\Phi$. 
\end{definition}

The following definition uses the concept of hyperreals introduced in the preliminaries section. We call a  vector $\phi \in (^* \mathbb R)^n$ of hyperreal weights    \emph{simple} if there exists   a positive infinitesimal number   $\epsilon$ and a sequence $\Phi = (\phi^1,\phi^2,\dots,\phi^T)$ of nonnegative and eventually positive weight vectors  in $\R^n$ such that 
\begin{align}
    \phi = \sum_{t \in \{1,\dots, T\}} \epsilon^{t-1} \phi^t. \label{suwm-to-simple}
\end{align}
An example with two individuals can illustrate the restriction associated with a  ``simple'' hyperreal vector.  Consider the hyperreal weight vector $(1+\epsilon,1)$, with $\epsilon$ being a positive infinitesimal number.  This vector  is not simple. To see this, note that the only way to express the vector $(1+\epsilon,1)$ in the form $\phi = \sum_{t \in \{1,\dots, T\}} \epsilon^{t-1} \phi^t$ is to set $T=2$, $\phi^1=(1,1)$, and $\phi^2=(1,0).$ The sequence $(\phi^1,\phi^2)$ violates the eventual positivity requirement.
By contrast, the weight vector $(1+\epsilon,\epsilon)=(1,0)+\epsilon(1,1)$ is simple.  The relevance of the distinction between simple and nonsimple hyperreal vectors, as well as the role played by the former, will become  clear in \cref{social-preference-section}.

\begin{definition}[Simple hyperreal utilitarian welfare maximization (SHUWM)]\label{def:simple-hyperreal-welfare-function}
A social welfare function $W$ is a \emph{simple hyperreal utilitarian welfare function} if 
	\begin{align}
	W(u) = \langle \phi, u \rangle = \sum_{t \in \{1,\dots,T\}} \epsilon^{t-1} \langle  \phi^t,u \rangle, \label{hyperreal-welfare-function}
	\end{align} 
where $\phi$ is a simple hyperreal weight vector. We say $u \in U$ maximizes a simple hyperreal utilitarian welfare function over $U$---or, more simply, $u$ is a \emph{SHUWM solution} of $U$---if there exists a simple hyperreal utilitarian welfare function $W$ such that $W(u) \ge W(v)$ for all $v \in U$.
\end{definition}

We can now state the first main result of the paper. 

\begin{theorem}\label{theorem:characterize-PO}
Let $U$ be a closed convex subset of $\mathbb R^n$ and let $u$ be a vector in $U$. Then, the following are equivalent:
\begin{enumerate}[label=(\roman*)]
    \item $u$ is Pareto optimal with respect to $U$.
    \item $u$ is a SUWM solution of $U$.
    \item $u$ is a SHUWM solution of $U$.
\end{enumerate}
\end{theorem}
\begin{proof}
See \cref{sec:only-if}.
\end{proof}

In the remainder of the section, we will offer interpretations of this result and insights into its proof. 

We first start with the SUWM characterization of Pareto optimality implied by the equivalence between (i) and (ii). In SUWM, utilitarian welfare is maximized over multiple rounds for growing sets of agents until all agents are considered. From the  social choice perspective, one can imagine a utilitarian social planner  who  prioritizes some agents---that is, those  considered in earlier rounds of SUWM---and maximizes their (weighted) welfare before  others. To achieve Pareto optimality, the social planner must assign \emph{some} weights to all agents, but the welfare weights for some individuals (those who receive positive weights in later rounds) may need to be infinitely smaller than those for others (those who receive positive weights in the earlier rounds). SUWM allows such flexibility by placing positive weights on individuals in different rounds. The eventual positivity condition encodes the requirement of Pareto optimality that ``every agent's welfare counts'' since  the utility of each  agent  $i$ has a positive weight in some  round of welfare  maximization.

The equivalence between (i) and (ii) is easy to visualize with the example in \cref{fig:illustration-of-sequential-converse}, reproduced in \cref{fig:illustration-of-sequential-optimization}(a). 
\begin{figure}
  \centering
  \begin{subfigure}{.45\textwidth}
    \centering
    \includegraphics[scale=.9]{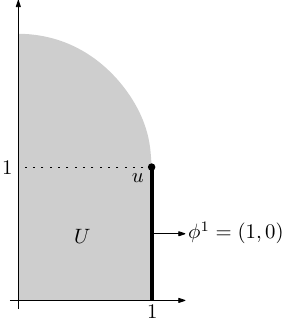}
    \caption{First round}
  \end{subfigure}%
  \hskip 20pt
  \begin{subfigure}{.45\textwidth}
    \centering
    \includegraphics[scale=.9]{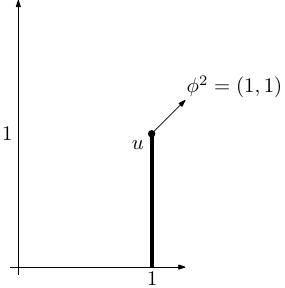}
    \caption{Second round}
  \end{subfigure}%
  \caption{Determining a Pareto optimal point in two rounds of sequential utilitarian welfare maximization.}
  \label{fig:illustration-of-sequential-optimization}
\end{figure}
In the first round, utilities are maximized within $U$ with weights $\phi^1$, which is maximized by the thick vertical segment containing $u$. One can interpret this as the social planner first maximizing the utility of agent 1 while disregarding the welfare of the other individual completely. Since agent 1 is indifferent among all of these points, the social planner seeks to engage in further optimization.  In the second (and last) round, utilities are again maximized but only within the vertical segment, now with an (arbitrary) nonnegative weight vector $\phi^2$ that places a positive weight on agent 2. Hence, $\Phi=(\phi^1, \phi^2)$ is eventually positive.  The weights $\phi^2$ determine $u$ as the unique maximizer, as illustrated in \Cref{fig:illustration-of-sequential-optimization}(b). The theorem shows that the flexibility in assigning the weights in different rounds in SUWM enables an exact characterization.

The equivalence of (ii) and (iii) in the theorem shows that the sequential optimization involved in SUWM can be encoded in a one-shot weighted utilitarian welfare maximization with the introduction of simple hyperreal weights. This introduction of hyperreals allows for some agents to be prioritized over others in the sense of being assigned positive weights in earlier rounds of SUWM. 
One can then interpret the former agents as carrying infinitely larger weights than the latter agents to constitute social welfare. The characterization in (iii) formalizes this idea by constructing the simple hyperreal weight vector $\phi =\sum_{t \in \{1,\dots,T\}} \epsilon^{t-1} \phi^t$. Since $\epsilon^s$ is infinitely larger than $\epsilon^t$ for any $t > s \ge 0$, the hyperreal vector $\phi$ assigns infinitely larger weights to the agents with higher priority than those with lower priority. For example, in \cref{fig:illustration-of-sequential-optimization} the vector $u = (1,1)$ maximizes the simple hyperreal utilitarian welfare function with hyperreal weights $\phi^1 + \epsilon \phi^2 = (1+\epsilon, \epsilon)$.\footnote{See Theorem 12.7 in \cite{soltan2015lectures}, reproduced as \Cref{fact-5} in the appendix.} While serving as a useful step toward our proof, this result lacks an important element that is fundamental in the economics context---that the weights be nonnegative and eventually positive.  A nontrivial and crucial part of our proof lies in showing that nonnegative and eventually positive weights can be found if and only if the face consists of Pareto optimal points. The proof of \cref{theorem:characterize-PO} in \cref{sec:only-if} provides additional details and discussion.


 Let us now explore some of the insights behind the proof of \cref{theorem:characterize-PO}. The argument showing that (i) implies (ii) exploits a remarkable parallel between our problem and the question in convex geometry pertaining to \emph{extreme faces} of a closed convex set.  An extreme face, or simply a {\it face,} $F$ of $U$ is its convex subset 
whose elements cannot be expressed as convex combinations of points outside that set. (An extreme point is a special case of a face comprised of a singleton.) Geometrically, Pareto optimal points of $U$ are made up of such  faces (a result we establish).  We say a hyperplane of $U$ ``exposes'' a face $F$ if it intersects $U$ precisely at  $F$, namely when $F$ constitutes the set of points that maximize a linear function.  A standard utilitarian welfare characterization of Pareto optima implies that the corresponding faces are  ``exposed'' by hyperplanes with  nonnegative weight vectors. From this perspective, the failure of standard weighted utilitarianism can be traced to the fact known in convex geometry that extreme faces may not always be exposed.  However, an important finding in that literature is that an extreme face is  ``eventually exposed,'' that is, the face can be represented by  the set of points that sequentially maximize  \emph{possibly negatively-}weighted sum of utilities.\footnote{See Theorem 12.7 in \cite{soltan2015lectures}, reproduced as \Cref{fact-5} in the appendix.}  While serving as a useful step toward our proof, this result lacks an element that is important for us and fundamental in the economics context---that the weights be nonnegative and eventually positive.  A nontrivial and crucial part of our proof lies in showing that nonnegative and eventually positive weights can be found if and only if the face consists of Pareto optimal points. 

 Next, the fact that (ii) implies (iii) follows since any SUMW solution constitutes a SHUWM solution with the simple hyperreal weights constructed using a sequence of the SUWM weights as in \eqref{suwm-to-simple}. 
Finally, we establish that (iii) implies (i), by observing that any SHUWM solution must be Pareto optimal, given the positivity of the simple hyperreal weights. 
%
%

\section{Axiomatic foundation for ``near'' utilitarianism}\label{social-preference-section}

In the previous section, we showed that ``near'' utilitarian welfare maximization---in the form of either SUWM or SHUWM---characterizes Pareto optima.  Here we provide an axiomatic foundation for these welfare criteria.  That is, we  identify axioms of welfare orderings represented by these social welfare criteria.

This exercise serves at least two purposes.  First,
one can view the preceding characterization (\cref{theorem:characterize-PO}) as providing a foundation for {\it some} version of utilitarianism.  It is important to ask exactly what social welfare ordering  corresponds to that version of utilitarianism.  Second, 
our version of utilitarianism relaxes standard utilitarianism by allowing for a sequence of welfare weights or for  hyperreal welfare weights in utilitarian welfare maximization. Identifying the social welfare orderings that justify such procedures will lay bare the precise nature of departure from those generating standard utilitarianism.  This difference will in turn make precise, and flesh out, the sense in which our utilitarianism is ``near'' the standard one. 

We begin with a state-of-the-art axiomatization of (weighted) utilitarianism. Let the social welfare ordering $R^*$ be a complete and transitive binary relation defined over $\mathbb{R}^{n}$,  the set of utility profiles of agents $I$, and let $P^*$ and $I^*$ denote the  strict  and indifferent parts of $R^*$, respectively.   For any $u\in \mathbb{R}^n$ and any real number $\delta>0$, let $B_{\delta}(u):=\{v\in \mathbb{R}^n:  ||v-u||<\delta\}$ be  the $\delta$-ball  centered at $u$.   Utilitarianism (with positive welfare weights) satisfies the following three axioms:   
\begin{itemize}
    
  \item \textsf{Pareto Principle}:   for any $u>v$, we have $u P^* v$.

    \item \textsf{Invariance}: for any $u,v\in \mathbb{R}^{n}$, $a \in \mathbb{R}^{n}$ and $b \in \mathbb{R}_{++}$, if $u R^* v$, then $(a+bu)  R^* (a+bv)$.

       \item  \textsf{Continuity}:  If $u P^* v$, then there exists $\delta>0$ such that 
       $u' P^* v$ for all $u'\in B_{\delta}(u)$.
\end{itemize}
 \textsf{Pareto Principle} requires the welfare ordering to preserve the Pareto domination order.
\textsf{Invariance} means that rescaling utility profiles by adding the same constant vector  or by multiplying with the same positive coefficient does not alter their  social welfare ordering. This property permits just the right scope of interpersonal utility comparison that yields linear social welfare evaluation.  \textsf{Continuity} means that perturbing the utilities of possibly all agents slightly does not alter social welfare ordering.  \textsf{Continuity} forces welfare weights on alternative individuals to be of the same order of magnitude at the margin, meaning that no individual is treated infinitely better or worse compared with the others. Theorem 4.2-(2) of \cite{d2002social} shows that utilitarianism is the only social welfare  ordering that satisfies the three axioms:

\begin{theorem}\label{thm:dasprement-gevers}
    [D'Asprement-Gevers] Let $R^*$ be a social welfare ordering. The following statements are equivalent:\footnote{Theorem 4.2-(1) of \cite{d2002social} gives the characterization with nonnegative welfare weights when Pareto is replaced with a weaker Pareto-like condition.}
    \begin{enumerate}[label=(\roman*)]
        \item $R^*$ satisfies the \textsf{Pareto Principle}, \textsf{Invariance,} and \textsf{Continuity},
        \item There exists $\phi\in \mathbb{R}^{n}_{++}$ such that 
 $u R^* v$  if and only if $\sum_{i\in I} \phi_i u_i\ge \sum_{i\in I} \phi_i v_i.$
    \end{enumerate} 
\end{theorem}   

It is easy to see that \textsf{Continuity} fails in our simple hyperreal utilitarian welfare function.  Recall
that in \cref{fig:illustration-of-sequential-optimization}, the Pareto optimum $u=(1,1)$ maximizes the simple hyperreal utilitarian welfare function $W(\cdot)$ with weights $(1 + \epsilon, \epsilon)$, where $\epsilon>0$ is an infinitesimal. Hence, $W(1,1)> W(1,1/2)$, for example. Yet,  for any real number $\delta>0$,
 $W(1-\delta, 1-\delta) < W(1,1/2)$, so $W$ fails \textsf{Continuity}.  Indeed, it is well-known that lexicographic preference orderings cannot be represented by a continuous utility function (see, for instance, pages 46-7 of \cite{mas1995microeconomic}).  

 While continuity in its general form cannot be satisfied, the additional structure of our {\it near} utilitarianism may accommodate some weaker version of continuity.  Indeed, we identify the precise form of weakening of  \textsf{Continuity} compatible with our near utilitarianism.  For each agent $i\in I$ and  a real number $\delta>0$, let $B_{\delta}^i(u):=\{v\in \mathbb{R}^n:  |v_i-u_i|<\delta, v_j=u_j, \forall j\ne i\}$ be the $\delta$-ball around $u$ but {\it only in the $i$-th coordinate}. This notion allows us to define:

\begin{itemize}
     \item \textsf{Weak Continuity}: for any $ u P^* v$, there exist  $i\in I$ and $\delta >0$  such that $u' P^* v$ for all $u'\in B_{\delta}^i(u)$.
\end{itemize}

\textsf{Weak Continuity} requires the social welfare ordering to be robust to perturbations  of only {\it some} individual agent's utility, and not necessarily to {\it all} possible  perturbations of the utility profile,  as required by \textsf{Continuity}.  We next present the desired axiomatization of our ``near'' weighted utilitarian welfare functions. To this end, we adapt SUWM to welfare orderings in a natural way.

 \begin{definition} We say $u$ \emph{sequentially  utilitarian welfare dominates} $v$ {\it according to} $\Phi$ if $u$  sequentially maximizes utilitarian welfare over $\{u,v\}$ 
 according to $\Phi$.\footnote{In words, $u$ sequentially  utilitarian welfare dominates $v$, if there exists a sequence of eventually positive weight vectors
  $\Phi = (\phi^1, \phi^2, \dots, \phi^T)$ satisfying: either  $\phi^t u= \phi^t v$ for all $t$ or  there exists $\tau \ge 1$ such that  $\langle \phi^t, u \rangle = \langle \phi^t, v \rangle$ for all $t  < \tau$ and $\langle \phi^{\tau}, u \rangle > \langle \phi^{\tau}, v \rangle$.}\footnote{Note that this ranking leads to   a \emph{rational} order, i.e.,  a binary relation that is reflexive, complete, and transitive. To see the transitivity (since the other properties are obvious),  consider profiles $u,v,$ and $w$ such that $u$ and $v$ sequentially utilitarian welfare dominate $v$ and $w$, respectively: that is,  $\langle  \phi^\tau, u \rangle >  \langle  \phi^\tau, v  \rangle$ for some $\tau$ and $\langle  \phi^t, u \rangle =  \langle  \phi^t, v  \rangle$ for all $t < \tau$ while $\langle  \phi^{\tau'}, v \rangle >  \langle  \phi^{\tau'}, w  \rangle$ for some $\tau'$ and $\langle  \phi^t, v \rangle =  \langle  \phi^t, w  \rangle$ for all $t < \tau$. Then, letting $\tau'' =\min \{ \tau, \tau' \}$, we have $\langle  \phi^{\tau''}, u \rangle >  \langle  \phi^{\tau''}, w  \rangle$ and $\langle  \phi^t, u \rangle =  \langle  \phi^t, w  \rangle$ for all $t < \tau''$, implying $u$ sequentially utilitarian welfare dominates $w$.}
\end{definition}

\begin{theorem}\label{social-preference-theorem} Let $R^*$ be a social welfare ordering. The following statements are equivalent.

\begin{enumerate}[label=(\roman*)]
    \item $R^*$ satisfies the  \textsf{Pareto Principle}. \textsf{Invariance,} and \textsf{Weak Continuity},
    \item There exists a nonnegative and eventually positive sequence of weight vectors $\Phi=(\phi^1, \phi^2, ..., \phi^T)$ such that for any $u,v \in \mathbb{R}^n$,  $u R^* v$ if and only if $u$ sequentially utilitarian welfare dominates $v$ according to $\Phi$. 
    
\item  
There exists a simple hyperreal weight vector $\psi\in ({}^* \mathbb R_{++})^n$ such that for any $u,v \in \mathbb{R}^n$, 
$u R^* v$ if and only if $ \sum_{i\in I} \psi_i u_i\ge \sum_{i\in I} \psi_i v_i.$\footnote{As with the order based on sequential utilitarian welfare domination, an order based on  this ranking  is also rational (as  hyperreal numbers  follow  the same ordering system as real numbers).}
\end{enumerate}
\end{theorem}
\begin{proof} See \cref{s:proof-of-social-preference-theorem}.
\end{proof}
 
For (iii), the restriction to {\it simple} hyperreal utilitarian welfare functions is crucial. Recall  that simplicity captures the eventual positivity of the  weight vectors  required in our SUWM, and this feature is essential for a hyperreal utilitarian welfare  function to retain the weak continuity property.  To see this, recall  the non-simple weight vector $\psi:=(1+\epsilon,1)$ with an infinitesimal $\epsilon>0$ discussed in \Cref{pareto-optimality-section}.   The welfare function associated with this weight vector fails \textsf{Weak Continuity}. To see this, consider utility profiles $u:=(1,0)$ and $v:=(0,1).$ We have $u P^* v$ because $\langle \psi, u \rangle=1+\epsilon>1=\langle \psi, v \rangle.$ However, for any  $i\in I$,  real number $\delta>0$, and $u' \in B^i_\delta(u)$ with $u' < u,$ we have $\langle \psi, u' \rangle<1=\langle \psi, v \rangle,$ so $u' P^* v$ does not hold, a violation of \textsf{Weak Continuity}.  The reason for this difference is that this non-simple hyperreal function cannot be supported by a nonnegative and eventually-positive sequence of weight vectors required by SUWM.

By contrast, consider the simple hyperreal utilitarian welfare function $W(\cdot)$  with weights $(1 + \epsilon, \epsilon)$ that exposes $u$ in  \cref{fig:illustration-of-sequential-converse}.   While $W$ fails to be continuous, it is weakly continuous.  Although $W$ fails \textsf{Continuity}, it satisfies \textsf{Weak Continuity}. Recall $W(u)>W(v)$, for $u=(1,1)$ and $v=(1,1/2)$. And, $W(u')>W(v)$ for any $u'\in B^2_{\delta}(u)$ if  $\delta\in (0, 1/2)$.

To visualize some of this discussion, \cref{fig:axiom-venn-diagram} illustrates the relationship between the axiomatizations of different notions of utilitarianism described in \cref{thm:dasprement-gevers,social-preference-theorem,hyperreal-swf-theorem} (the last result is discussed in the next section).   

\begin{figure}[htb]
	\centering
	\includegraphics[scale=.9]{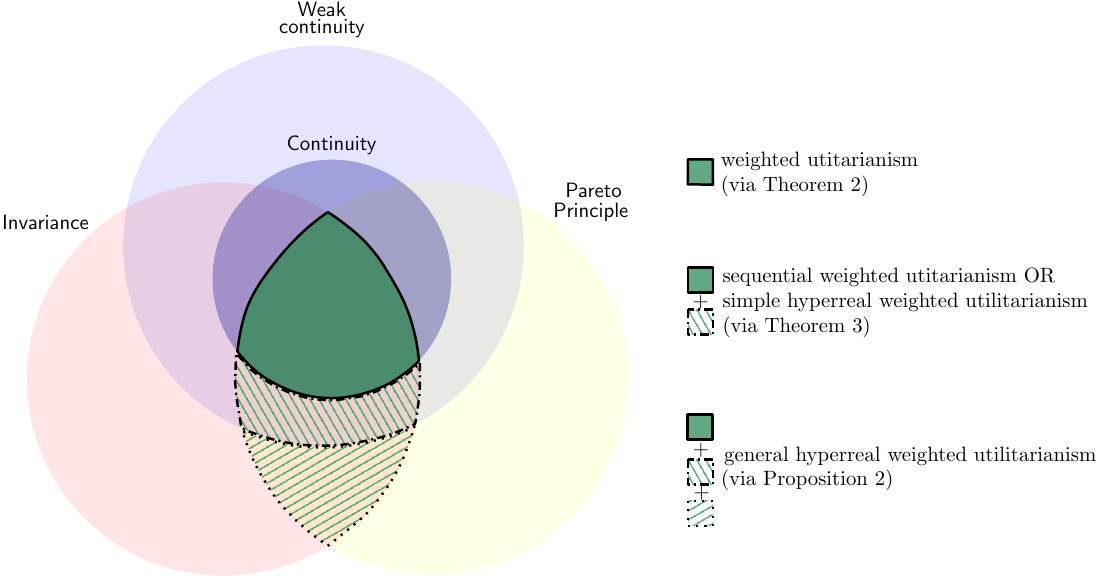}
	\caption{Illustrating the axiomatizations of different notions of utilitarianism. The universe is the set of all social welfare orderings. 
 }
		\label{fig:axiom-venn-diagram}
\end{figure}

\section{Other characterizations of Pareto optimality}\label{s:discussion}

In this section, we discuss other characterizations of Pareto optima.  As will be seen, these characterizations  are not only related to  our ``near''-utilitarian welfare maximizations but they also highlight certain aspects of them and thus  help to interpret and understand them.  At the same time, we will show that they differ in their axiomatic properties from our ``near''-utilitarian welfare characterizations.  Our discussion will therefore illustrate that the axiomatic properties of our ``near''-utilitarian characterizations are special and not shared by other possible characterizations of Pareto optima.

\paragraph{Weighted utilitarianism with general hyperreal weights.}

As we discussed, the restriction to {\it simple} hyperreal utilitarian welfare functions disciplines them to resemble utilitarianism.  At the same time, hyperreal utilitarian welfare maximization, with no restriction, also characterizes Pareto optimality.

\begin{proposition}\label{prop:hyperreal}
Let $U$ be a closed convex subset of $\mathbb R^n$. Then, $u \in U$ is Pareto optimal if and only if 
\begin{equation*}
u \in \argmax_{u' \in U} \langle \psi,u' \rangle,
\end{equation*}
for some weight vector $\psi=(\psi_i)_{i \in I}\in ({}^* \mathbb R_{++})^n$.

\end{proposition}
\begin{proof}
See \cref{sec:hyperreal-proof} in the Supplementary Appendix. 
\end{proof}

This proposition highlights the ability to assign an infinitely larger weight to one agent relative to another as a crucial feature that enabled SHUWM to characterize Pareto optimality.  Compared with simple hyperreal utilitarian welfare functions, however, the class of general hyperreal utilitarian welfare functions is too large to be declared near-utilitarian.  As we already saw, the class includes non-simple hyperreal functions that do not satisfy \textsf{Weak Continuity}, let alone \textsf{Continuity}.  

\Cref{theorem:characterize-PO} tells us that such non-simple functions are not needed for characterizing Pareto optimality. To illustrate their superfluity, recall the non-simple hyperreal vector $\psi=(1+\epsilon, 1)$, where $\epsilon$ is a positive infinitesimal number.  We can see that any such vector can be replaced  by a simple hyperreal vector (which does satisfy \textsf{Weak Continuity}), in this particular case, a real vector, with no loss on the ability to characterize Pareto optima.\footnote{In this case, the real weight vector $(1,1)$ can be used in place of $\psi$ in the sense that every  Pareto optimal point that maximizes the hyperreal  weight  vector $\psi$ also maximizes the real weight vector $(1,1)$.} We showed that the welfare function associated with  $\psi=(1+\epsilon, 1)$   fails \textsf{Weak Continuity}.
Indeed, the next proposition shows that  the class of hyperreal utilitarian welfare functions in \Cref{prop:hyperreal} entails no restriction on social welfare orderings beyond the \textsf{Pareto Principle} and \textsf{Invariance}.
\begin{proposition} \label{hyperreal-swf-theorem} Let $R^*$ be a social welfare ordering. The following statements are equivalent.
       \begin{enumerate}[label=(\roman*)]
    \item  $R^*$ satisfies the  \textsf{Pareto Principle}  and \textsf{Invariance}. \label{item-1}
\item There exists a hyperreal weight vector $\psi\in    ({}^* \mathbb R_{++})^n$ such that 
$u R^* v$ if and only if $ \sum_{i\in I} \psi_i u_i\ge \sum_{i\in I} \psi_i v_i.$  
\label{item-2} 
\end{enumerate}
\end{proposition}
\begin{proof}
See \cref{ss:proof-of-hyperreal-swf-theorem} in the Supplementary Appendix.
\end{proof} 
\cref{fig:axiom-venn-diagram} illustrates the differences in how  general hyperreal utilitarian and other utilitarian welfare functions are axiomatized.

\paragraph{Sequential Nash bargaining.}   The second characterization is motivated by an institutional/behavioral implementation of Pareto optima. As is well known from 
the second fundamental  welfare theorem, a Pareto optimal allocation, say in an exchange economy,  can be implemented by a competitive equilibrium under a suitable endowment.%
\footnote{As an aside, in  \cref{s:second-welfare-thm} in the SupplementaryAppendix, we illustrate how to use some of the techniques established in our proof of \cref{theorem:characterize-PO} to offer a new proof  of the second welfare theorem that allows for weaker assumptions  than  the standard treatment.  We discuss this more in the paper's conclusion section.}
In the same spirit, one may ask what institution implements a given Pareto optimum in a more general environment. Our SUWM characterization of Pareto optima allows one to envision sequential negotiations as fulfilling this goal. That is, any Pareto optimal outcome can be seen as emerging from a sequence of negotiations among individuals whose relative bargaining powers in round $t$ are determined by the welfare weights  $\phi^t$ in the corresponding round of SUWM characterization.

To be specific, suppose each agent has a disagreement utility, normalized as zero, that is less than any Pareto optimal utility---i.e., $u\gg0$ for every $u\in U^P$.   Consider a collection of  bargaining units $\mathcal I=\{I^1, \dots, I^T\}$ satisfying $I^{t-1} \subsetneq I^t$  for each $t=2,\dots,T$  and $I^T =I$. Imagine that the agents engage in a sequence of bargaining: in round $1$, agents in $I^1$ bargain from $U$ to a set $V^1\subset U$, and in round $t=2,\dots,T$, agents in set $I^t$ bargain from $V^{t-1}$ to a set $V^t$. The bargaining protocol in each round $t$ is a generalized Nash bargaining game \citep{kalai1977nash} in which each agent $i\in I^t$ has a bargaining   power  $\psi_i^t >0$ such that $\sum_{i\in I^t}\psi_i^t=1$ and a disagreement payoff $0$. More specifically, for bargaining units $\mathcal{I} =\{I^1, \dots, I^T\}$ and bargaining powers $\Psi  = (\psi^1,\ldots, \psi^T)$ satisfying the above requirement, 
 we let $V^t := \arg\max_{u \in V^{t-1}} \prod_{i\in I^t}  u_i^{\psi_i^t}$ for each $t =1,\ldots,T$ with $V^0: = U$. Then, we call any $u \in V^T$  a {\it sequential Nash bargaining solution} (SNBS) over $U$ for $\mathcal{I}$ and $\Psi$, and call  $u$ an SNBS over $U$ if there exist such $\mathcal{I}$ and $\Psi$.

Observe now that SNBS implements the SUWM procedure for the logarithmic transforms of utilities.  Namely,  $u$ is an SNBS  over   $U\subset \mathbb{R}_{++}^n$ if and only if $v:=(\ln u_1, ..., \ln u_n)$ is an SUWM solution of $V:=\{(\ln u_1', ..., \ln u_n'): (u_1', ..., u_n') \in U\}$. This connection also makes  it clear that SNBS provides another characterization of  Pareto optima.

\begin{proposition}\label{cor:nash-bargaining} A vector $u \in U \cap \mathbb{R}^n_{++}$ is Pareto optimal if and only if  $u$ is an SNBS over $U$.
\end{proposition}
\begin{proof}
See \cref{s:nash-bargaining} in the Supplementary Appendix.
\end{proof}

This result provides a  behavioral interpretation of our near-weighted utilitarian welfare maximization.  Despite this close connection, we will see that the SNBS characterization differs in the social welfare ordering it induces from our near-weighted utilitarian characterizations. To see this, we first define the welfare ordering induced by SNBS. We say $u$ \emph{sequentially Nash welfare dominates $v$} according to bargaining units $\mathcal{I}$ and  bargaining powers $\Psi$ if $u$ is an  SNBS over $\{ u, v\}$ for  $\mathcal{I}$ and  $\Psi$.

Since SNBS implements the SUWM procedure for the logarithmic transforms of utilities, \Cref{social-preference-theorem}  implies that the following axiom would fulfill the same role as \textsf{Invariance}. 

\begin{itemize}
    \item \textsf{Log Invariance:} for any $u, v \in \mathbb{R}^n_{++}$, if $u R^* v$, then  $u' R^* v'$ for any $u',v' \in \mathbb{R}^n_{++}$ such that, for some $a \in \mathbb{R}^n$ and $b \in \mathbb{R}_{++}$,  $\ln u'_i =a_i + b \ln u_i $ and $\ln v'_i  = a_i + b \ln v_i$ for all  $i\in I$.  
\end{itemize}

Combining this axiom with the  \textsf{Pareto Principle} and \textsf{Weak Continuity} defined earlier, we obtain the following axiomatization of the welfare ordering based on SNBS.%

\begin{corollary} \label{SNBS-axiomatization} Let $R^*$ be a social welfare ordering defined on $\mathbb{R}^n_{++}$. Then, the following statements are equivalent.
\begin{itemize}
    \item[(i)] $R^*$ satisfies the \textsf{Pareto Principle}, \textsf{Log Invariance}, and \textsf{Weak Continuity}.
    \item[(ii)]   There exist  bargaining units $\mathcal{I}$ and bargaining powers $\Psi$ such that for any $u,v \in \mathbb{R}^n_{++}$, $u R^* v$ if and only if $u$ sequentially Nash welfare dominates $v$ according to $\mathcal{I}$ and $\Psi$.  
\end{itemize} 
\end{corollary}
\begin{proof}
See \cref{s:snbs-axiomatization} in the Supplementary Appendix.
\end{proof}

In particular, this corollary implies that while SNBS characterizes Pareto optimality, the welfare orderings implied by the criterion depart further from utilitarianism than our ``near''-utilitarian welfare criteria. While it shares the  \textsf{Pareto Principle} and \textsf{Weak Continuity}, it generally fails \textsf{Invariance}.

 \paragraph{Piecewise linear concave welfare function.} Some readers may not like the sequentiality of SUWM or the use of hyperreal numbers in SHUWM. This observation leads to the question of whether it is possible to characterize Pareto optima by a one-shot maximization of a real-valued welfare function. For such a characterization, the welfare function cannot be weighted utilitarian. In particular, the function must be nonlinear. Can we achieve the characterization with minimal relaxation of the  linearity? This motivates the following approach.
 
A social welfare function $W$ is a \emph{piecewise linear concave (PLC) welfare function} characterized by  $(\psi^1,\psi^2,\dots,\psi^t)$  if  
\begin{align}\label{eq:plc_functon} 
W(v) = \min_{t\in \{1,\ldots,T\}} \langle \psi^t, v \rangle,
\end{align} 
where $\psi^t \in \mathbb{R}^n_+$ for each $t$. One candidate for the weight vectors $(\psi^1,\psi^2,\dots,\psi^T)$ to construct  a PLC welfare function are those identified in the SUWM characterization; i.e., eventually positive weights. However,  the characterization  does \emph{not} hold without an auxiliary condition. For this condition, let us  say that a PLC welfare  function $W$   \emph{achieves its maximum over $U$ via eventually positive weights} if       (i) $ (\psi^1,\psi^2,\dots,\psi^T)  $ is nonnegative and eventually positive and  (ii)     for all  $v \in \arg\max_{u'\in U} W (u')$, $W(v) = \langle \psi^T, v \rangle $. 

\begin{proposition} \label{prop:PLC-char}
Let $U$ be a closed convex subset of $\mathbb R^n$. Then, $u \in U \cap \mathbb{R}_{++}^n$ is Pareto optimal if and only if  it maximizes a PLC welfare  function that achieves its maximum over $U$ via eventually positive weights.\footnote{We focus on points $u \in \R^n_{++}$ for technical simplicity. This is not a substantive restriction because the economic environment is arguably unchanged when a constant is added to all utility profiles.
}\footnote{This proposition may be reminiscent of construction of a PLC utility function based on an individual's choice data (see \cite{afriat1967}).   
The PLC social welfare function reveals the planner's preferences for agents' utilities similarly to how Afriat's PLC utility function reveals an individual's preferences for alternative goods.
Note, however, that there are clear differences. The multiple linear components of our PLC welfare function result from multiple welfare weights corresponding to the successive rounds of SUWM.  By contrast, the linear components in  Afriat's construction reflect different budget lines a consumer faces in different choice scenarios. 
Moreover, the role played by the auxiliary condition to ensure every agent's welfare counts has no analogue in Afriat's characterization.} 
\end{proposition}
\begin{proof}
See \cref{sec:plc-cha-proof} in the Supplementary Appendix.
\end{proof}
The role of the auxiliary condition is to prevent a Pareto suboptimal point from maximizing the PLC function (so that the ``if'' direction holds).  To see it, observe that for any Pareto suboptimal point $u$, one can find $v > u$ so that $W(v) \ge W(u)$.  If $u$ were a maximizer of $W$, then the auxiliary condition would  require $W (v) = \langle  \psi^T, v \rangle = \langle  \psi^T, u \rangle = W (u) $ or $\langle  \psi^T,  v-u \rangle =0$, which cannot hold since $\psi^T \gg 0$ and $v > u$.      
While achieving the goal of characterizing Pareto optima,  the auxiliary condition also captures  the main feature   of SUWM that every agent's welfare must count as it requires a PLC function to be maximized via a weight vector that puts a positive weight on every agent's utility.

While our PLC welfare functions successfully characterize Pareto optima, we regard them to be further away from utilitarianism than our ``near''-utilitarian welfare criteria. This is because, to our knowledge, no natural axioms characterize PLC welfare functions. In fact, it is not even obvious how to formulate a PLC function as a social welfare ordering in the face of the auxiliary condition. For instance, define the binary relation $R^*$  by $u R^* v$  if $W(u) \ge W(v)$ and $W(u)=\langle \psi^T,u \rangle$. Note that the condition $W(u)=\langle \psi^T  ,u\rangle$ is an adaptation of the auxiliary condition to the context of social welfare ordering.
Then, $R^*$ is not necessarily a complete binary relation, as the
following example shows. 

\begin{example}
Let there be two agents $1$ and $2$, $T=2$, $  \psi^1 =(1,0)$, $ \psi^2 =(1,1)$, $u=(1,1)$ and $v=(0,0)$. Then, we have $W(u)=1>0=W(v)$ while  $W(u)=1<2= \langle  \psi^2 ,u \rangle$, so  neither $u R^* v$ nor $v R^* u$ holds. Hence,  $R^*$ is not complete.\end{example}

\section{Conclusion}\label{s:conclusion}

We have provided two characterizations of Pareto optimal solutions of a closed convex set that are ``near'' to weighted utilitarian maximization in an axiomatic sense.  They arise from relaxing the \textsf{Continuity} axiom that defines weighted utilitarian to a \textsf{Weak Continuity} axiom. We have shown that other characterizations of Pareto optimality are more ``distant'' from weighted utilitarianism because they violate more of its defining axioms.  These results constitute significant progress in clarifying the connection between Paretian and utilitarian notions that are foundational to welfare economics.
 
 Although our paper directly worked with the space of utility profiles $U$, our results drive implications for problems stated in the choice space $X$. Indeed, examining the structure of what points in the choice set give rise to Pareto optima has been a major focus in the multiobjective optimization literature. An early contribution in that literature is \cite{charnes1967management}, who showed an  equivalence between the problem of finding Pareto optimal solutions (in the choice set $X$) 
and that of solving a constrained nonlinear programming problem.  Following their contribution, techniques in nonlinear programming were utilized to characterize Pareto optima under various conditions \citep{ben1977necessary,van1994saddle,glover1999dual,ben1980characterization} all of which require some form of differentiability of the utility functions. We believe further investigation into our approach may have the potential to add to this literature in at least two aspects. First, our characterization does not assume any form of differentiability. 
 Indeed, the subtlety of non-exposure of Pareto optimal faces can also arise when utility functions are not smooth, as is often the case.   
Our methods may suggest ways to handle Pareto optimality when differentiability fails. Second, our methods may suggest a bridge between existing results in the choice space and results in the utility possibility space, where notions of (sequential) welfare maximization are salient and allow for more natural economic interpretations. Indeed, none of the characterizations in the above references speak to notions of welfare maximization.

A second area of future work would be to examine how the notion of exposure  can be used to enhance separating hyperplane arguments that may arise in other economic settings. For instance, the second welfare theorem relies on the existence of a  strictly positive weight vector for  a  supporting hyperplane (which  constitutes equilibrium prices).  One can prove this  with a weaker assumption than in the existing proof of the theorem by  leveraging the idea of exposing a Pareto optimal point---which is a target Pareto efficient allocation---, as we show in the Supplementary Appendix (see \cref{s:second-welfare-thm}).   
We believe there is scope to explore other economic settings where separating hyperplane arguments are used and similarly relax the conditions needed to ensure strict positivity when Pareto optimality (in combination with notions of exposure) may be used to assure the existence of a separating hyperplane with a positive weight vector. 

A third area of future work is to extend the characterization presented in this paper to the case of infinite-dimensional economies. This is not a straightforward extension. Our argument in the finite-dimensional case depends on a termination condition that counts dimension. In the infinite-dimensional case, this termination condition is not accessible to us. Generalization would likely require a set convergence argument (for instance, using Hausdorff or Kuratowksi set-based metrics) that avoids discussion of dimension.

\appendix 


\section{Appendix: Proof of \cref{theorem:characterize-PO}}\label{sec:only-if}

\subsection{Proof of (ii) $\Rightarrow$ (iii)}
Given the sequence $\Phi= (\phi^1, \ldots, 
\phi^T)$ that is sequentially maximized by $u$,  let us construct a simple hyperreal weighted welfare function $W(\cdot)$ as in \eqref{hyperreal-welfare-function}. Letting $U^0,U^1,\dots,U^T$ be the notation used in \Cref{def:optimizing-set-of-normals},  observe first that 
	\begin{align}\label{inequality-1}
	W(v)	 =\langle \phi, v \rangle=\langle \phi, u \rangle   = W (u) \text{ for every $v \in U^T$},
	\end{align}  by construction of $\phi$ and $U^T$. Next, consider any $v \not \in U^T$. Then, there exists $\tau \in \{1,\dots,T\}$ such that $\langle \phi^t, v \rangle=\langle \phi^t, u \rangle$ for all $t<\tau$ and $\langle \phi^{\tau}, v \rangle < \langle \phi^{\tau}, u \rangle$. Therefore,
	\begin{align}
		\langle \phi, u \rangle -  \langle \phi, v \rangle & = \sum_{t=1}^T  \epsilon^{t-1} \langle \phi^t,u \rangle - \sum_{t=1}^T  \epsilon^{t-1} \langle \phi^t,v \rangle  \nonumber  \\
		&= \epsilon^{\tau-1} \left (\langle \phi^{\tau},u \rangle -\langle \phi^{\tau},v \rangle \right ) + \sum_{t \in \{\tau+1,\dots,T\}}  \epsilon^{t-1} \left (\langle \phi^{\tau},u \rangle -\langle \phi^{\tau},v \rangle \right ) \nonumber \\
		&= \epsilon^{\tau -1} \left  [\left (\langle \phi^{\tau},u \rangle -\langle \phi^{\tau},v \rangle \right ) + \sum_{t \in \{\tau +1,\dots,T\}}  \epsilon^{t-\tau} \left (\langle \phi^{\tau},u \rangle -\langle \phi^{\tau},v \rangle \right ) \right ] \label{eq:square-bracket}
	\end{align} 
	where all arithmetic operations are valid because $^* \mathbb R$ is an ordered field (Theorem 3.6.1 of \cite{goldblatt2012lectures}).
	Because $\epsilon$ is an infinitesimal number strictly larger than zero, and $\langle \phi^{\tau},u \rangle-\langle \phi^{\tau},v \rangle$ is a   positive real number, the expression inside the square bracket in \cref{eq:square-bracket}  is   positive, so  $  \langle \phi, u \rangle -  \langle \phi, v \rangle >0$ (see pages 50 and 51 of \cite{goldblatt2012lectures}). Thus, we have shown that 
	\begin{align}\label{inequality-2}
	W (v) = 	\langle \phi, v \rangle <  \langle \psi, u \rangle    = W(u) \text{ for every $v \not\in U^T$}.
	\end{align}
	By equations \Cref{inequality-1,inequality-2}, we have established that $u \in \arg\max_{v \in U}  W (v).$

\subsection{Proof of  (iii) $\Rightarrow$ (i)}  Suppose for contradiction  that $u$ maximizes the simple hyperreal weighted welfare function $W$  as in \eqref{hyperreal-welfare-function} but   is not Pareto optimal. Then, there exists $v \in U$ such that $v> u$.  Since the sequence $(\phi^1, \ldots, \phi^T)$ is nonnegative and eventually positive, we have  $ \phi_i = \sum_{t \in  \{1, \ldots, T\}}\epsilon^t   \phi_i^t >0$ for each $i \in I$.   Thus,  $W(v)  - W(u)= \langle \phi,v \rangle-\langle \phi,u \rangle=\sum_{i \in I} \phi_i (v_i-u_i)>0$, a contradiction.

\subsection{Proof of (i) $\Rightarrow$ (ii)}
  We begin with some preliminaries before providing the proof in \Cref{s:proof}. 

\subsubsection{Preliminaries}

Let us first introduce a few concepts that are crucial for our  analysis. A \emph{face} of  $U$ is a nonempty convex subset $F$ of $U$ with the property that if $u \in F$ and $u = \alpha v + (1 - \alpha) w$ for some $0 < \alpha < 1$ and $v, w \in U$ then it must be that $v, w \in F$. That is, $F$ is a face of a convex set if none of its elements are convex combinations of elements that lie outside of $F$.  A \emph{proper face} of $U$ is a face of $U$ that  is a proper subset of $ U$. 
A face  $F$ is  an \emph{exposed face} of $U$ if there is a weight vector $\phi \in \R^n$  such that  $F = \arg\max_{u \in U} \langle \phi, u \rangle$. In this case, we say that $\phi$ exposes $F$ out of $U$. A face need not be exposed, as can be seen  in \cref{fig:illustration-of-sequential-converse}, where $u$ is a singleton face that is not exposed.  

For any convex subset $G$ of $U$, its \emph{relative interior} $\text{ri}(G)$ is the set of all $u \in G$ such that for every $u' \in G$ there exists  $\lambda > 0$ such that $u + \lambda (u - u') \in G$.  
 
The following lemma shows a face structure of a convex set that is interesting in itself and useful for our analysis.
\begin{lemma}[Corollary 11.11(a) in \cite{soltan2015lectures}]\label{lemma:rockafellar-union}
For a convex set $U \subseteq \R^n$,  the collection of relative interiors of faces---that is, $\{\mathrm{ri}(F) : F \mbox{ is a  face of  } U \}$---forms a partition of $U$. 
\end{lemma}

The next lemma shows that Pareto optimal points ``come in faces.'' It is standard in the convex analytic literature to refer to Pareto optimal points as \emph{maximal} points, so we use that language here.

\begin{lemma}\label{lemma:maximal-face}
Suppose a maximal point $u$  of   a closed convex set $U$ lies in the relative interior of a face $F$ of $U$.  Then, every point in $F$ is maximal. \end{lemma} 
\begin{proof}
The stated result  is immediate in the case  $F$ is a singleton, so we may assume that $F$ is not a singleton. Suppose for contradiction that  $F$ contains  a nonmaximal element $u'$. Thus, there exists a $v \in U$ such that $v > u'$.  Since $u \in  \mathrm{ri} (F)$, there exists  $\lambda > 0$ such that $w' = u + \lambda (u -u') \in F$.  Now let $z= \alpha  w' + (1- \alpha) v$, where $\alpha  =\frac{1}{1+\lambda}$ or $\alpha (1+\lambda) =1$. Note that $z \in U$ since $U$ is convex. Moreover, 
\begin{align*}
z  =  \alpha \left(u + \lambda (u -u') \right)  + (1-\alpha) v  \nonumber 
     = u - \alpha \lambda u'  + (1-\alpha) v     = u + (1-\alpha) (v -u') > u, 
\end{align*}
contradicting the maximality of $u$.
\end{proof}

According to \Cref{lemma:maximal-face}, we say a face is \emph{maximal} if all of its elements are maximal.  
Importantly for our purpose, \cref{lemma:rockafellar-union,lemma:maximal-face} imply that every maximal point of $U$ belongs to a relative interior of a unique maximal face of $U$ (possibly $U$ itself).   
 
The next result provides a key step of our argument: every face, possibly non-exposed, is \emph{eventually} exposed.%
%
%
\footnote{Theorem 5 of \cite{lopomo2019} proves the same result for singleton faces $F$, i.e., extreme points.} 
%
%
 
\begin{lemma}[Theorem 12.7 in \cite{soltan2015lectures}] \label{fact-5}
Let $U \subset \mathbb{R}^n$ be a convex set and $F$ be a nonempty proper face of $U$.  There is a sequence of convex sets $(G^t)_{t=0}^T$ such that
\begin{align*}
F  =G^T \subset G^{T-1} \subset \cdots \subset G^1 \subset G^0 =U,
\end{align*} 
where $G^t$ is a nonempty proper exposed face of $G^{t-1}$ for each $t =1, \ldots, T$.
\end{lemma}

This lemma is already illustrated in the Introduction. In \cref{fig:illustration-of-sequential-optimization}, the singleton face $u$ is exposed in two rounds: the vertical segment is exposed first by a weight vector $(1,0)$, and then $u$ is exposed by weight vector $(1,1)$  (among many others) out of that vertical segment. This lemma is not enough for our result, however, as it is silent about any additional properties on the weight vectors that expose the sequence of faces. Crucially,  our characterization requires the weight vectors to be nonnegative and eventually positive.  

For these additional features, we need to introduce a set of analytical tools.  Let $J$ be any subset of the index set $I$ and let $\chi^J$ denote the vector whose $i$-th coordinate is equal to $1$ for every $i \in J$  and equal to $0$ for every $i \notin J$. When $J$ is the singleton $\{i\}$ we simplify $\chi^{\{i\}}$ to $\chi^i$. A convex set $U$ is \emph{downward closed in coordinates $J \subset I$  }  if,   for all $u \in U$ and all $\tau \ge 0$,    $u - \tau \chi^K \in U$ for any subset $K$ of $J$. A convex set that is downward closed in all coordinates $I$ is simply called \emph{downward closed}.  The \emph{downward closure} of a closed convex set   $U$   is the downward closed set $\dc(U) := \bigcup_{u \in U} (u - \R^n_+)$. It is straightforward to see that  $\text{dc}(U)$ is closed and convex if $U$ is closed and convex. 

One useful feature of downward closure is that it preserves maximal elements and thus maximal faces.

\begin{lemma}\label{lemma:downward-closure-suffices}
The set of maximal elements of a closed convex set coincides with that of its downward closure.
If $F$ is a maximal face of $U$ then $F$ is a maximal face of $\mathrm{dc} (U)$. 
\end{lemma}
\begin{proof}
Let $U$ be a closed convex set and $\dc(U)$ its downward closure. Let $u$ be a maximal element of $\dc(U)$; that is, $(u + \R^n_+) \cap \dc(U) = \{u\}$. If $u \in U$ then this implies $(u + \R^n_+) \cap U = \{u\}$ since $U \subset \dc(U)$ and so $u$ is a maximal element of $U$. Note that if $u \in \text{dc}(U) \setminus U$ then it cannot be maximal. Indeed, this implies that $u = v - w$ for some $v \in U$ and nonzero $w \in \R^n_+$ and so $v > u$ and so $u$ is not maximal. 

Conversely, we prove the contrapositive. Suppose $u \in \dc(U)$ is not a maximal element. This implies that there exists a $w \ne u$ with $w \in \dc(U)$ and $w \ge u$. However, then we can find a $v \ge w \ge u$ and $v \ne u$ and $v \in U$. This implies that $u$ is not a maximal element of $U$. 
We next prove the second statement. To see that $F$ is a face of $\mathrm{dc} (U)$, consider any $x,y \in \mathrm{dc} (U)$ and $\lambda \in (0,1)$ such that  $z= \lambda x   + (1-\lambda) y \in F$.  We need to show that  both $x$ and $y$ belong to $F$. We first show that $x$ and $y$ are both maximal. Suppose for contradiction that $x$ is not maximal. Then, we must have some $x' \in \mathrm{dc} (U) $ such that   $x' > x$. Let  $z' = \lambda x' + (1-\lambda) y $ and observe that $z' \in \mathrm{dc} (U)$,   $z' \ge z $, and $z' \ne z$, which contradicts the maximality of $z$.  Given that $x$ and $y$ are both maximal, we must have $x,y \in U$ since there is no maximal point in $\mathrm{dc} (U) \backslash U $. That $F$ is a face of $U$  then implies $x,y \in F$ as desired. 
\end{proof}

Crucially for our arguments, halfspaces of the form $\{ u : \langle \phi, u \rangle \le W_\phi\}$  that contain  downward-closed sets must have nonnegative weight vectors.

\begin{lemma}\label{lemma:nonnegative-supporting-hyperplanes}
Let $U$ be a set that is downward closed in coordinates $J \subset I$.
 If $U$ is contained in the halfspace $\{ u \in \R^n : \langle \phi, u \rangle \le W_\phi\}$, then  $\phi_j \ge 0,\forall j \in J$. 
\end{lemma}
\begin{proof}
Suppose for contradiction that $\phi_j<0$ for some $j \in J$. Let $v$ be an arbitrary element of $U$. Since $U$ is downward closed in coordinates $J$, we also have $v - \lambda \chi^j \in U$ for any $\lambda \ge 0$, where $\chi^j$ is the unit vector with $1$ in component $j$. However, observe that $\langle \phi, v - \lambda \chi^j \rangle = \langle \phi, v \rangle - \lambda \langle \phi, \chi^j \rangle = \langle \phi, v \rangle - \lambda \phi_j$. But  $\langle \phi, v \rangle - \lambda \phi_j \to \infty$ as $\lambda \to \infty$ since $\phi_j<0$. This contradicts the fact that is contained $\{ u \in \R^n : \langle \phi, u \rangle \le W_\phi\}$.
\end{proof}

\begin{lemma}\label{lemma:exposed-faces-of-downward-closure-are-partially-downward-closed}
Let $F$ be a  face of a closed convex set $U$   that is downward closed in coordinates $J \subset I$.  If $\phi$  exposes $F$ out of $U$, then  $F$ is downward closed in coordinates  $J \setminus \supp \phi$. 
\end{lemma}
\begin{proof}
Take any $j \in K := J \setminus \supp \phi$ and set $u ' = u - \epsilon \chi^j$ for some $u \in F$ and $\epsilon > 0$. Since $U$ is downward closed in coordinates $J$ and $j \in J$, we have $u' \in U$. Moreover, $\langle \phi, u' \rangle = \langle \phi, u - \epsilon \chi^j \rangle = \langle \phi, u \rangle - \epsilon \langle \phi, \chi^j \rangle = \langle \phi, u \rangle - \epsilon \phi_j = \langle \phi, u \rangle$ since $\phi_j = 0$ when $j \in K$ since no element of $K$ lies in $\supp \phi$. However, then $u' \in F$ since $\langle \phi, u' \rangle = \langle \phi, u \rangle = \max_{v \in U} \langle \phi, v \rangle$ and $F = \arg\max_{v \in U} \langle \phi, v \rangle$ since $F$ is exposed by $\phi$. 
\end{proof}

\subsubsection{Proof of (i) $\Rightarrow$ (ii)}\label{s:proof}

Fix any maximal point $u$ of $U$.  We wish to show that $u$ sequentially maximizes utilitarian welfare  over $U$.     The proof consists of several steps.

\begin{step} \label{step1}  There exists a unique face $F$ of $\text{dc}(U)$ such that $u\in \mathrm{ri}(F)$.  All points of $F$ are maximal in $\text{dc}(U)$. 
\end{step}

\begin{proof}  By \cref{lemma:downward-closure-suffices},  $u$ is a maximal point  of $\mathrm{dc}(U)$. By \cref{lemma:rockafellar-union} there is a unique face $F$ of $\mathrm{dc}(U)$ which  contains $u$ in $\mathrm{ri}(F)$. By \cref{lemma:maximal-face}, every point of $F$ is maximal in $dc(U)$, as desired.
\end{proof}

\begin{step}  \label{step2} The face $F$ (containing $u$) is a proper face of $\text{dc}(U)$.  
\end{step}

\begin{proof}  If not, we must have $F=\text{dc}(U)$.  Pick any $u'\in  \text{dc}(U)$.  Then, for any $\epsilon>0,$ $u''  = u' - \epsilon \chi^I$  is also in $\text{dc}(U)$ by the downward closure property. Clearly,  $u''$ is not a maximal point of $\text{dc}(U)$ and cannot belong to $F$ by \cref{step1}, a contradiction. 
\end{proof}

 \begin{step}  \label{step3} There exists  a sequence of convex sets $(G^t)_{t=0}^T$ of $\text{dc}(U)$ such that $G^t$ is a proper exposed face of $G^{t-1}$ for  $t = 1, \ldots, T$, where $G^0 = \mathrm{dc}(U)$, $G^T = F$, and $ T\le n$.
 \end{step}
 
\begin{proof} Since $F$ is a  proper face of  $\text{dc}(U)$ by \cref{step2}, the result follows from   \cref{fact-5}. For any set $V$, let $\text{dim} (V)$ denote its dimension.%
%
%
\footnote{The dimension $\mathrm{dim}(V)$ of a convex subset $V$ of $U$, including one of $U$'s faces, is defined by the dimension of its affine hull: $\mbox{aff}(V):= \{\sum_{j=1}^k\alpha_j v^j \mid k\in \mathbb{N},  v^j\in V, \alpha_j\in \mathbb{R}, \sum_{j=1}^{k} \alpha_j=1 \}$.}
%
%
If $V'$ is a proper face of convex set $V$, then   $\text{dim}(V') < \text{dim} (V)$ by Theorem 11.4 in \cite{soltan2015lectures}. Thus,     we have $ T\le n$ since $\text{dim}(G^{t})<\text{dim}(G^{t-1})$ and since $\text{dim}(G^0) = \text{dim}(\mathrm{dc}(U))=n$.
 \end{proof}

 \begin{step}  \label{step4} There exists  a sequence  $\Phi=(\phi^1, \dots,
 	\phi^T)$ such  that  for each $t=1,\dots,T$,
\begin{equation*}
G^{t}=\arg\max_{x\in G^{t-1}} \langle \phi^t, x\rangle,
\end{equation*}
where $\phi^t >0$ and $ \supp \phi^1 \subset \supp \phi^{2} \subset \ldots \subset \supp \phi^T = I $.\footnote{Note that $\Phi$, with these properties, is eventually positive since $\phi^T \gg 0$.}
\end{step}

\begin{proof}  By \cref{step3}, there exists  a sequence of weight vectors $\Psi=(\psi^1, \dots, \psi^T)$ such that, for each $t=1,\dots,T$,  $\psi^t$ exposes $G^t$ out of $G^{t-1}$.  We construct  $\Phi= (\phi^1, \dots, \phi^T)$ with the stated properties.  
	
The construction is recursive.  First, since $G^0 = \mathrm{dc}(U)$, by  \cref{lemma:nonnegative-supporting-hyperplanes},  $\phi^1:= \psi^1$ is nonnegative. For an inductive hypothesis, suppose that there are $\phi^k$, $k=1,\dots,t-1$, with the stated properties and that for each $k = 1,\dots,t-1$,  $G^{k}$ is downward-closed in  coordinates  $J^{k} := \{i \in I \mid \phi^{k}_i=0\}=I\setminus \supp \phi^{k}$. Note that $J^{t-1}\subset J^{t-2}\subset \cdots \subset J^{0}:= I$.  We will now construct $\phi^t$  and show that $G^t$ is downward-closed in coordinates $J^t = \{i \in I \mid \phi^{t}_i=0\}$. 
	
First, observe   $G^{t-1}$ is contained in $\{ u :  \langle   \psi^t, u \rangle \le \max_{ u' \in G^{t-1} }  \langle   \psi^t, u' \rangle  \}$ and  $G^{t-1}$ is downward-closed in coordinates $J^{t-1}$.  Hence,   \cref{lemma:nonnegative-supporting-hyperplanes} implies that  $\psi^t_j\ge 0$ on coordinates $j\in J^{t-1}$.  Consider next  $i\in \supp \phi^{t-1}=I\setminus J^{t-1}$.  For such $i$, it is indeed  possible for $\psi^t_i$ to be negative. But   noting $\phi^{t-1}_i>0$ for such $i$, we define
\begin{equation*}
   \phi^t = \lambda^t  \phi^{t-1}+ \psi^t,
\end{equation*}
 where $\lambda^t >  \max_{i \in \supp \phi^{t-1}}  |\psi^{t}_i |/\phi^{t-1}_i$ is a (sufficiently large) positive scalar.  Given this construction, $\phi^t_i\ge 0$ for all $i\in I$ and  $\phi^t_i> 0$ for all $i\in \supp \phi^{t-1}$; i.e., $\supp \phi^t \supset \supp \phi^{t-1}$.
	 	
	 Let us  show that $\phi^t$ exposes $G^t$ out of $G^{t-1}$. To this end, let $M^t:=  \max_{x\in G^{t-2}} \langle \phi^{t-1}, x\rangle$.   For all $x\in G^{t-1}$, we have
\begin{equation*}
\langle \phi^t , x \rangle   =  \lambda^t \langle  \phi^{t-1} , x\rangle  + \langle \psi^t, x \rangle  =\lambda^t M^t   + \langle \psi^{t}, x \rangle,
\end{equation*}
	 since $\langle \phi^{t-1}, x\rangle =M^t$ for all $x\in G^{t-1}$.  Henceforth, 
\begin{equation*}
\arg\max_{x\in G^{t-1}} \langle \phi^t , x \rangle = \arg\max_{x\in G^{t-1}}  \langle \psi^{t}, x \rangle= G^{t}.
\end{equation*}
	Since $G^{t-1}$ is downward-closed in coordinates $J^{t-1}$ and $\phi^t$ exposes $G^t$ out of $G^{t-1}$, \cref{lemma:exposed-faces-of-downward-closure-are-partially-downward-closed} implies that $G^t$ is downward-closed in coordinates $J^{t-1} \setminus \supp  \phi^t =  (I \setminus \supp \phi^{t-1}) \setminus \supp  \phi^t = I\setminus \supp  \phi^t = J^t $, where the penultimate equality holds since $\supp \phi^{t-1} \subset \supp \phi^t$.

It remains to show that for each $i \in I$, there exists $t \in \{1,\dots,T\}$ such that $\phi^t_i>0$. To show this, it suffices to show that $\phi^T \gg 0$. Supposing not, there must be  some $i \in I$ such that   $\phi^t_i=0$ for all $t =1,\dots,T$, so $i \in J^t$ for all $t =1, \dots,T$.    Then,  \cref{lemma:exposed-faces-of-downward-closure-are-partially-downward-closed} implies that for all $t =1,\dots,T$,  $G^t$ is downward-closed in coordinate $i$, which  contradicts the fact that $G^T = F$ is maximal.
\end{proof}

We have so far shown that $u$ sequentially maximizes welfare   over $\text{dc}(U)$.  We now prove the main result: $u$ sequentially   maximizes welfare  over $U$.  To this end, the following last step suffices. 

\begin{step}\label{step5} $u$ sequentially maximizes utilitarian welfare  over $U$.
\end{step}

\begin{proof} Recall a sequence of weight vectors $\Phi$ from \cref{step4}.  Let  $U^0, U^1,\ldots, U^T$ be   convex subsets of $U$ such that, for each $t=1,\dots,T$, $U^t$ is the face of $U^{t-1}$ exposed by weight vector $\phi^t$; i.e.,
\begin{equation*}
U^t =\arg\max_{x\in U^{t-1}}\langle \phi^t, x\rangle,
\end{equation*}
	where $U^0:= U$.  It suffices to prove that $U^T=F$, as this will prove that $u$ sequentially maximizes utilitarian welfare  over $U$.   
	
	To this end, it suffices to prove that $F \subset U^t \subset G^t$ for each   $t = 0,\ldots,T$.  We proceed inductively for the proof. First, note that the claim is trivially true for $t=0$ because $U^0 := U \subset dc(U) := G^0$  and $F \subset U = U^0$ by definition. Now, suppose that the claim holds for $t$. We show (i)  $F \subset U^{t+1}$  and (ii) $U^{t+1} \subset G^{t+1}$ as follows.
	
	For (i), fix any point $v$ in $F$. Then, since $F \subset G^{t+1}$ and $\phi^{t+1}$ exposes $G^{t+1}$ out of $G^t$, we have $\langle \phi^{t+1},v \rangle \ge \langle \phi^{t+1},w \rangle$ for every $w \in G^t$. Because $U^t \subset G^t$ by the inductive assumption, 
	\begin{equation}\label{eq:maximal-here}
	\langle \phi^{t+1},v \rangle  \ge \langle \phi^{t+1},w \rangle
	\end{equation}
	for every $w \in U^t$. Moreover, $v \in U^t$ by the assumption that $F \subset U^t$. This fact, combined with \cref{eq:maximal-here}, implies that $\phi^{t+1}$ is maximized by $v$ over $U^t$ and so $v \in U^{t+1}$, since $\phi^{t+1}$ exposes $U^{t+1}$ out of $U^t$. This holds for every $v \in F$ and so $F \subset U^{t+1}$, implying (i) holds for $t+1$.
	
	As for (ii), fix any point $v$ in $U^{t+1}$. By (i), we know that 
	\begin{equation}\label{eq:first-part}
	\langle \phi^{t+1},v \rangle = \langle \phi^{t+1},w \rangle
	\end{equation}
	for any $w \in F$, since $U^{t+1}$ is exposed by $\phi^{t+1}$ and $F$ is a subset of $U^{t+1}$.  Moreover, by the definition of $G^{t+1}$ and the fact that $F \subset G^{t+1}$ by construction, we know that 
	\begin{equation}\label{eq:second-part}
	\langle \phi^{t+1},w \rangle  \ge \langle \phi^{t+1},z \rangle
	\end{equation}
	for any $w \in F$ and $z \in G^t$. Combining \cref{eq:first-part} and \cref{eq:second-part} implies that $\langle \phi^{t+1},v \rangle  \ge \langle \phi^{t+1},z \rangle$  for any $z \in G^t$. This, and the fact that $v \in G^t$ (which immediately follows from $v \in U^{t+1} \subset U^t \subset G^t$), means that $v \in G^{t+1}$. Since this holds for any $v \in U^{t+1}$, we can conclude that $U^{t+1} \subset G^{t+1}$, so (ii) holds for $t+1$. 
	
	This completes the induction and establishes the result. \end{proof}
	
\section{Appendix: Proof of \cref{social-preference-theorem}}\label{s:proof-of-social-preference-theorem}

\subsection{Proof of (i) $\Rightarrow$ (ii)} 
We first show that if $R^*$ satisfies the \textsf{Pareto Principle} and \textsf{Invariance}, then there exists a sequence $\Phi  = (\phi^1, \phi^2, \ldots, \phi^T)$  of nonnegative and  eventually positive weight vectors such that for any $u, v \in \mathbb{R}^n$,  $u P^* v$ if $u$ sequentially utilitarian welfare dominates $v$, that is,    \begin{align}
    \label{i-to-ii-1}  \langle  \phi^t, u \rangle  >  \langle \phi^t , v \rangle \mbox{ for some $t$} \, \mbox{ and } \,  \langle  \phi^s, u \rangle  = \langle \phi^s , v \rangle \mbox{ for all } s < t.  
\end{align} 
We will then show if $R^*$ satisfies \textsf{Weak Continuity} in addition, then \cref{i-to-ii-1} implies $u P^* v$. Thus,   $u P^* v$  if and only if  $u$ sequentially utilitarian welfare dominates $v$.   

To do so,  define $S :=\{ s \in \mathbb{R}^n :  0 R^* s \}$ and 
 $Q := \{ s + p :   s \in S,   p \in \mathbb{R}^n, \mbox{ and } p \ll 0 \}. $   We now observe that 
$Q$ is convex.  To this end, let $q, q'\in Q$ and $q'' =t q+(1-t)q'$ for any $t\in (0,1)$.   Then, $q=s+p$ and $q'=s'+p'$ for some $s,s'\in S$ and $p,p'\ll 0$, and  $q''=t s+ (1-t) s'+ t p+(1-t)p'$. Then, by \textsf{Invariance}, we have $ 0 R^*  t s$,  $0  R^* (1-t)s'$,  and   $  (1-t) s'  R^*   t s+(1-t)s'$. Thus, by transitivity,  we have  $0 R^* ts +(1-t)s' $, i.e., $ts +(1-t)s' \in S$. Since  $ t p+(1-t)p'\ll 0$, we have $q'' \in Q$.

 Letting $\bar Q$ denote the closure of $Q$,   $\bar Q$ is also convex. Note  that $0 \in  S \subset \bar  Q$ and that by the \textsf{Pareto Principle}, $0$ is a maximal point of both $S$ and $\bar Q$.  Also, there is a maximal face $F \subset S$ with $0 \in F$. Letting $G^0 := \bar Q$,  the  same proof as    Step 4 in the proof of \cref{theorem:characterize-PO} can be used to show  there exists  a sequence   $\Phi=(\phi^1, \dots,
 	\phi^T)$ such  that  for each $t=1,\dots,T$,
\begin{equation}
G^{t}=\arg\max_{x\in G^{t-1}} \langle \phi^t, x\rangle, \label{eq:seqential-max}
\end{equation}
where $\phi^t >0$, $\supp \phi^t \supsetneq \supp \phi^{t-1}$, $\phi^T \gg 0$, and $G^T = F$.

\begin{claim}  \label{claim-1}
    \cref{i-to-ii-1} implies $u P^* v$
\end{claim}  
\begin{proof}
    Suppose for contradiction that there are some $u,v$ for which \cref{i-to-ii-1} holds   but  $v R^* u$.   Let $w := u -v$. Then, by \textsf{Invariance},    $ 0 R^*w  $, so $w \in S \subset \bar Q =G^0$. By the hypothesis, we have  $\langle  \phi^s, w \rangle =0,\forall s < t$.  Since $ 0 \in F  = G^T$, \cref{eq:seqential-max} implies  $w \in G^s,\forall s < t$, which in turn implies $\langle  \phi^t, w \rangle  \le   \langle \phi^t ,  0 \rangle =0$, or $ \langle  \phi^t, u \rangle  \le   \langle \phi^t , v \rangle$. We thus have a contradiction.
\end{proof} 

\begin{claim}
  $u P^* v$ implies \cref{i-to-ii-1}    
\end{claim}
\begin{proof} Let us  first prove that  $\langle \phi^s, u \rangle =\langle \phi^s, v \rangle, \forall s$ implies   $u I^* v$.  Suppose for a contradiction that  $u P^* v$.  By \textsf{Weak Continuity},  there are $i$ and $\delta >0$ such that   $ u' P^* v$ for all $u' \in B_{\delta}^i (u)$. We  can then find a round $t$ in which  $i \in  \supp   \phi^t  \backslash \supp \phi^{t-1}$. Since $\langle  \phi^s ,  \chi^i \rangle  =0,\forall s< t$ and $\langle \phi^t, \chi^i  \rangle >0$ for the unit vector $\chi^i$ whose $i$-th component is equal to 1,
we have  $ \langle \phi^s, u- \delta' \chi^i \rangle = \langle \phi^s, u \rangle =\langle \phi^s, v \rangle, \forall s < t$ and $\langle \phi^t, u-\delta' \chi^i \rangle < \langle \phi^t, u \rangle = \langle \phi^t, v \rangle$ for any $\delta' >0$, which implies by the former statement that $v P^* (u-\delta' \chi^i) $, contradicting that $u' P^* v $ for all $ u' \in B_{\delta}^i (u)$.

Since  $\langle \phi^s, u \rangle =\langle \phi^s, u \rangle, \forall s$ implies $u I^* v$,    $u P^* v$  implies that  there must be some $t$ such that $\langle  \phi^s, u \rangle =  \langle \phi^s, v \rangle,\forall s < t $ and $\langle  \phi^t, u \rangle  \ne   \langle \phi^t, v \rangle $. Since $\langle  \phi^t, u \rangle   <   \langle \phi^t, v \rangle $ would imply  $ v  P^* u$  by  \cref{claim-1},  we must have $\langle  \phi^t, u \rangle  >    \langle \phi^t, v \rangle $ as desired.  \end{proof}

\subsection{Proof of  (ii) $\Rightarrow$ (iii)}  Consider the sequence $(\phi^1,\phi^2, \ldots,\phi^T)$ in  (ii).  Then, $u P^* v$ is equivalent to \cref{i-to-ii-1}, which  implies  
\begin{align}
\sum_{i\in I} \psi_i u_i  - \sum_{i\in I} \psi_i v_i = \epsilon^{t-1} \Big[  \langle \phi^t, u -v \rangle  + \sum_{s > t} \epsilon^{s-t} \langle \phi^{s},u- v \rangle \Big] >0,   \label{hyperreal-comp}
\end{align} where the inequality holds since the first term in the square bracket is a positive real  and the second term is infinitesimal. 
Conversely, if   $\sum_{i\in I} \psi_i u_i > \sum_{i\in I} \psi_i v_i$, then there must be some $s$ such that $ \langle \phi^s, u  \rangle  >  \langle \phi^s, v \rangle.  $ Letting $t$ be the smallest  such $s$, we must have $ \langle \phi^r, u  \rangle  =  \langle \phi^r, v \rangle, \forall r < t  $: else if $\langle \phi^r, u  \rangle   <   \langle \phi^r, v \rangle$ for some $r < t$, then one can use a similar argument to \cref{hyperreal-comp} to obtain $\sum_{i\in I} \psi_i v_i >  \sum_{i\in I} \psi_i u_i$, a contradiction.

\subsection{Proof of (iii) $\Rightarrow$ (i)}   That  the welfare function in \cref{hyperreal-welfare-function}---or a social welfare ordering it represents---satisfies the \textsf{Pareto Principle} and \textsf{Invariance}  is  straightforward to check. 
To check that it satisfies \textsf{Weak Continuity}, consider any $u P^* v$ so that $W (u) > W(v)$. As argued before, there must be  some $t$ such that $ \langle \phi^t, u  \rangle  >  \langle \phi^t, v \rangle$ and  $ \langle \phi^s, u  \rangle  =  \langle \phi^s, v \rangle, \forall s < t  $. Pick any $i \in  \supp   \phi^t  \backslash \supp \phi^{t-1}$. For sufficiently small $\delta >0$ and all $u' \in B^i_{\delta} (u)$, we have $ \langle \phi^t,  u' \rangle  >  \langle \phi^t, v \rangle$  while $\langle \phi^s, u' \rangle  =  \langle \phi^s, v \rangle,\forall s < t$. This implies as desired  that for all $u' \in B^i_{\delta} (u) $,
\begin{align*}
    W (u') - W (v)  = \epsilon^{t-1} \Big[  \langle \phi^t, u' -v \rangle  + \sum_{s > t} \epsilon^{s-t} \langle \phi^{s}, u' - v \rangle \Big] >0,   
\end{align*}  where the inequality holds for the same reason as \cref{hyperreal-comp} holds.

\bibliographystyle{economet}
\bibliography{bibliography-maximal}


\clearpage
\pagenumbering{arabic}
\renewcommand*{\thepage}{SA.\arabic{page}}

\begin{center}
    \huge{Supplementary Appendices for: \vskip 2pt ``Near'' weighted utilitarian characterizations of Pareto optima}
\end{center}

\section{Proofs for alternate characterization of Pareto optimality in \cref{s:discussion}}

\subsection{Proof of  \cref{prop:hyperreal}} \label{sec:hyperreal-proof}
	\emph{The ``only if'' direction.} Suppose $u \in U$ is Pareto optimal. Then, by \Cref{theorem:characterize-PO}, 
	there is a sequence $\Phi = (\phi^1, \phi^2, \dots, \phi^T)$ of  $T\le n$ nonnegative and eventually positive weight vectors  such that $u$ sequentially maximizes $\Phi$. Defining  $\psi:=\sum_{t \in \{1,\dots,T\}} \epsilon^{t} \phi^t$,    we have $\psi_i >0$ for each $i \in I$ since the vectors  $\Phi = (\phi^1, \phi^2, \dots, \phi^T)$ are  nonnegative and eventually positive. Also,  we have   $u \in \arg\max_{v \in U} \langle \psi, v \rangle$ by \cref{theorem:characterize-PO}. 
	\vskip 10pt
	\noindent \emph{The ``if'' direction.} To show the contrapositive, assume that $u$ is not Pareto optimal. Then there exists $v \in U$ such that $v> u$. Then, for any weight vector $\psi=(\psi_i)_{i \in I}$ with $\psi_i \in {^*\mathbb R}$ and $\psi_i>0$ for each $i \in I$, $\langle \psi,v \rangle-\langle \psi,u \rangle=\sum_{i \in I} \psi_i (v_i-u_i)>0$.  This means that  $u \not\in \argmax_{u' \in U} \langle \psi,u' \rangle$, as desired.

\subsection{Proof of \cref{hyperreal-swf-theorem}}\label{ss:proof-of-hyperreal-swf-theorem}

The \cref{item-2} $\Rightarrow$ \cref{item-1} direction is obvious. To prove  \cref{item-1} $\Rightarrow$ \cref{item-2},
we adopt the proof approach of Theorem 1' of an unpublished work by \citet{blume-unpublished} who studies an individual's decision under uncertainty.\footnote{\cite{blume-unpublished} is superseded by the published version,  \citet{blume1991lexicographic}, although our proof is more closely related to the former.} Suppose that the social welfare ordering $R^*$ satisfies the \textsf{Pareto Principle} and \textsf{Invariance}.

\begin{lemma}\label{blume-lemma8}
Let $U'=\{u^1,u^2,\dots,u^m\} \subset \mathbb R^n$ be a finite subset of utility profiles such that $u P^* \tilde u$ for some $u,\tilde u \in U'$. Then, there exists a nonnegative and non-zero weight vector $\phi^{U'} \in \mathbb R_+^n$   such that, for any $u, \tilde u \in U'$, $u R^* \tilde u$ if and only if $\langle \phi^{U'},u \rangle \ge \langle \phi^{U'},\tilde u \rangle$.  
\end{lemma}
\begin{proof}
We utilize the following fact:
\begin{lemma}[Lemma 7 of \citet{blume-unpublished}] \label{blume-lemma7}
Let $v^1,\dots, v^K$ and $w^{K+1},\dots,w^L$ be vectors in $\mathbb R^n$. Then, one of the following two statements holds.  
\begin{enumerate}
    \item \label{case-1} There exists $x \in \mathbb R_+^n \setminus \{0\}$ such that 
    \begin{align*}
        \langle x, v^k \rangle & >0 \text{ for all $k \in \{1,\dots,K\}$, and}\\
        \langle x, w^\ell \rangle & =0 \text{ for all $\ell \in \{K+1,\dots,L\}$}.
    \end{align*}
    \item \label{case-2}  There exist $y \in \mathbb R_+^K, z \in \mathbb R^{L-K}$ such that
    \begin{align*}
        \sum_{k=1}^K y_k v^k + \sum_{\ell=K+1}^L z_\ell w^\ell \le 0.
    \end{align*}
Moreover, if $y=0$, then     $\sum_{\ell=K+1}^L z_\ell w^\ell \neq 0.$
\end{enumerate}
\end{lemma}
Let $v^k, k=1,\dots, K$ be the vectors of the form $v^k=u^k-\tilde u^k$  where $u^k, \tilde u^k \in U'$ and $u^k P^* \tilde u^k,$ and $w^\ell, \ell=K+1,\dots, L$ be the vectors of the form $w^\ell=u^\ell-\tilde u^\ell$ where $u^\ell, \tilde u^\ell \in U'$ and $u^\ell I^* \tilde u^\ell.$ It suffices to show  that the case \ref{case-1} of \Cref{blume-lemma7} holds, as then the conclusion of \Cref{blume-lemma8} holds when we set the solution $x$ for the case \ref{case-1} of \Cref{blume-lemma7} as $\phi^{U'}$. To show this, we will show that the case \ref{case-2} of \Cref{blume-lemma7} does not hold. Suppose to the contrary that  the case \ref{case-2} of \Cref{blume-lemma7} holds, with solution $(y,z)$. We will obtain a contradiction.

Let $\tilde z_\ell:=|z_\ell| $ and $\tilde w^\ell:=sgn( z_\ell)  w^\ell$.
Then $(y,\tilde z) > 0$ and $\sum_{k=1}^K y_k v^k + \sum_{\ell=K+1}^L \tilde z_\ell \tilde w^\ell \le 0.$\footnote{To see why $(y,\tilde z) > 0$, note that if $(y,\tilde z) = 0$, then $y=0$ and $\sum_{\ell=K+1}^L z_\ell w^\ell = 0,$ a contradiction to case \ref{case-2} of \Cref{blume-lemma7}.} 
Because  $v^k$ for each $k=1,\dots,K$ is of the form  $u^k-\tilde u^k$ with $u^k P^* \tilde u^k$ while $\tilde w^\ell$ for each $\ell=K+1,\dots,L$ is of the form $u^\ell-\tilde u^\ell$ with $u^\ell I^* \tilde u^\ell$, it follows that 
\begin{align} \label{ineq}
      \sum_{k=1}^K y_k \tilde u^k + \sum_{\ell=K+1}^L \tilde z_\ell\tilde u^\ell \ge  \sum_{k=1}^K y_k u^k + \sum_{\ell=K+1}^L \tilde z_\ell u^\ell.
    \end{align}
 Since $R^*$ satisfies the \textsf{Pareto Principle},   this implies that 
\begin{align}
         \left ( \sum_{k=1}^K y_k \tilde u^k + \sum_{\ell=K+1}^L \tilde z_\ell \tilde u^\ell \right  ) R^* \left ( \sum_{k=1}^K y_k u^k + \sum_{\ell=K+1}^L \tilde z_\ell u^\ell \right  ).\label{social-preference-relation}
    \end{align}

Meanwhile, since $u^k P^* \tilde u^k$ for each $k$ and $u^\ell I^* \tilde u^\ell$ for each $\ell$ by assumption, by repeated applications of \textsf{Invariance},\footnote{Specifically, for any $u,v, \tilde u, \tilde v \in \mathbb R^n$ with $u R^* v$ and $\tilde u R^* \tilde v$ as well as  $\alpha,\beta \in \mathbb R_+$, we have $(\alpha u + \beta \tilde u) R^* (\alpha v + \beta \tilde v)$, with $R^*$ replaced with $P^*$ if $u P^* v$ and $\alpha>0$. This is because  $(\alpha u + \beta \tilde u) R^* (\alpha v + \beta \tilde u)$ and $(\alpha v + \beta \tilde u) R^* (\alpha v + \beta \tilde v)$ by \textsf{Invariance}, and hence by transitivity of $R^*$, the desired relationship follows (and the relation being strict if $u P^* v$ and $\alpha>0$).} it follows that \[\left ( \sum_{k=1}^K y_k u^k + \sum_{\ell=K+1}^L \tilde z_\ell u^\ell \right ) P^* \left ( \sum_{k=1}^K y_k \tilde u^k + \sum_{\ell=K+1}^L \tilde z_\ell \tilde u^\ell \right ),\]  if there exists $k$ with $y_k>0$,  a contradiction to \cref{social-preference-relation}. If $y_k=0$ for all $k$, then $\sum_{\ell=K+1}^L z_\ell w^\ell \neq 0$ by assumption, so we have by \cref{ineq},
\begin{align*}
          \sum_{\ell=K+1}^L \tilde z_\ell u^\ell < \sum_{\ell=K+1}^L \tilde z_\ell \tilde u^\ell,
          \end{align*}
holds.\footnote{To see this, suppose for contradiction that \cref{ineq} holds as equality. Then it would imply $\sum_\ell \tilde z_\ell \tilde w^\ell  = \sum_\ell z_\ell w^\ell =0 $, where the  first equality follows from the definitions of $\tilde z_\ell$ and $\tilde w^\ell$.} Since $R^*$ satisfies the \textsf{Pareto Principle}, this implies that 
\begin{align}
          \left ( \sum_{\ell=K+1}^L \tilde z_\ell \tilde u^\ell \right ) P^*           \left ( \sum_{\ell=K+1}^L \tilde z_\ell u^\ell \right ).
          \label{social-preference-relation-2}
    \end{align}

Meanwhile, recalling again $u^\ell I^* \tilde u^\ell$ for each $\ell=K+1,..., L$, and applying \textsf{Invariance} repeatedly, we have that 
\[\left ( \sum_{\ell=K+1}^L \tilde z_\ell u^\ell \right ) I^* \left ( \sum_{\ell=K+1}^L \tilde z_\ell \tilde u^\ell\right ), \] 
a contradiction to \cref{social-preference-relation-2}. This completes the proof.
\end{proof}

Now we proceed to complete the theorem. To do so, we define 
\begin{align*}
    \mathcal U:=\{U' \subset \mathbb R^n : |U'|<\infty, \exists u, v \in U', u P^* v\},
\end{align*}
and, for each $u \in \mathbb R^n,$ define the collection $U^u \subset \mathcal U$ by 
\begin{align*}
    \mathcal U^u:=\{U' \in \mathcal U  : u \in U'\}.
\end{align*}
Then, we consider a family $V$ of collections defined by
\[
V:=\{\mathcal U^u : u \in \mathbb R^n\}.
\]
Now, let $u^1,u^2,\dots,u^m \in \mathbb R^n$ and consider \[ \bigcap_{k=1}^m \mathcal U^{u^k}.\]
Note that  $ \bigcap_{k=1}^m \mathcal U^{u^k} \neq \emptyset$ since $\{u^1,u^2,\dots,u^m\} \in U^{u^k}$ for all $k \in \{1,\dots,m\},$ that is, the family $V$ has the finite intersection property.

Now, we invoke the following fact:
\begin{lemma}[Proposition 3.6 of \citet{joshi1983introduction}]
A collection of sets has the finite intersection property if and only if there is a filter that is a superset of that collection. 
\end{lemma}
The preceding argument and the claim imply that there exists a filter that is a superset of $V.$ By Zorn's lemma, there exists an ultrafilter $\Omega$ that is a superset of the above filter, and hence a superset of $V.$ This ultrafilter is clearly free, that is, the intersection of all sets in the collection $\Omega$ is empty: This is because all sets of the form $\mathcal U^u$ is an element of $\Omega$, and $\cap_{u \in \mathbb R^n}\mathcal U^u =\emptyset$.\footnote{To show $\cap_{u \in \mathbb R^n}\mathcal U^u =\emptyset$, suppose for contradiction that $\cap_{u \in \mathbb R^n}\mathcal U^u$ is nonempty, so there exists $U' \in \cap_{u \in \mathbb R^n}\mathcal U^u.$ Then, by definition $U'$ is a finite subset of $\mathbb R^n$. So there exists $v \in \mathbb R^n$ such that $v \not\in U'$. This implies $U' \not\in U^v,$ so $U' \not\in \cap_{u \in \mathbb R^n}\mathcal U^u$, a contradiction.}

Now, consider the set of functions from $\mathcal U$ to $\mathbb R.$ We say that two functions $r$ and $s$ are equivalent if $\{U' \in \mathcal U: 
 r(U')=s(U')\}$ is in $\Omega.$ It is straightforward to show that this is an equivalence relation and the set $^* \mathbb R$ of those equivalence classes is an ordered field which extends $\mathbb R.$\footnote{See \cite{blume-unpublished} for proofs of this property as well as others in this paragraph.}
 We call $^* \mathbb R$ the set of hyperreal numbers. It is well known that addition and multiplication defined on $\mathbb R$ extend  readily to $^*\mathbb R$ by pointwise operations, while the orders $\ge$ and $>$ also extend in a similar manner. It is also standard to show that there exists an infinitesimal number in $^* \mathbb R$.

Now,  let $\psi \in (^*\mathbb R)^n$ be such that, for each $i \in I$, $\psi_i$ is  the equivalence class that contains the element $r_i$ such that $r_i(U')=\phi^{U'}_i$ for each $U' \in \mathcal U$ and  $\phi^{U'}$ given in \Cref{blume-lemma8}. 
Consider any $u,v \in \mathbb R^n.$ We know that $\mathcal U^u \cap \mathcal U^v \in \Omega$. For any $U' \in \mathcal U^u \cap \mathcal U^v$, $u, v \in U'$, so if $u R^* v,$ then $\langle \phi^{U'}, u \rangle \ge \langle \phi^{U'}, v \rangle$. By construction of $\psi,$ this implies that $\langle \psi, u \rangle \ge \langle \psi, v \rangle$. A similar argument shows that $u P^* v$ implies $\langle \psi, u \rangle > \langle \psi, v \rangle$. Finally, for each $i \in I$, note that $\chi^i P^* 0$ as $R^*$ satisfies the \textsf{Pareto Principle}. Therefore, it follows that $\psi_i = \langle \psi, \chi^i \rangle > \langle \psi, 0 \rangle=0$, showing that $\psi \in (^* \mathbb R_{++})^n.$ This completes the proof.

\subsection{Proof of \cref{cor:nash-bargaining}}\label{s:nash-bargaining}

\begin{proof}
 For any $u \in \mathbb R_{++}^n$, let $\log u:=(\log u_i)_i$ and, moreover, for any $u \in \mathbb R^n$, let $e^u:=(e^{u_i})_i$. Let us also redefine $U:=U \cap \mathbb R^n_{++}$ for notational simplicity. Now, let
$
\tilde U:=\{\log u | u \in U\}
$. 

\begin{claim}\label{simple-observation}
Suppose $u \in U$ and let $\tilde u=\log u$. Then, $u$ is Pareto optimal with respect to $U$ if and only if $\tilde u$ is Pareto optimal with respect to $\dc(\tilde U).$ 
    \end{claim}
\begin{proof}
  First, note that $u \in U$ is Pareto optimal with respect to $U$ if and only if $\tilde u$ is Pareto optimal with respect to $\tilde U$ because $\log(.)$ is a strictly increasing function. Second, note that $\tilde u \in \tilde U$ is Pareto optimal with respect to $\tilde U$ if and only if it is Pareto optimal  with respect to $\dc(\tilde U)$ because Pareto optimality is invariant to adding utility vectors to a set that are smaller than existing utility vectors. These two observations imply the conclusion of this claim.
\end{proof}

\begin{claim}
Suppose that $U$ is convex. Then $\dc(\tilde U)$ is convex.
\end{claim}
\begin{proof}
    Suppose $\tilde u, \tilde u' \in \dc(\tilde U),$ and $\lambda \in [0,1].$ By definition of $\dc(.)$, it follows that there exist $\tilde v, \tilde v' \in \tilde U$ such that $\tilde u \le \tilde v, \tilde u' \le \tilde v'.$ Therefore, by definition of $\tilde U,$ there exist $v, v' \in U$ such that $\tilde v=\log v, \tilde v' = \log v'.$

    Because $U$ is convex, $w:=\lambda v + (1-\lambda) v'$ is in $U$. This implies that $\tilde w:=\log w$ is in $\tilde U.$ Now, because $\log(.)$ is a concave function, we have that
    $$
    \lambda \tilde v+(1-\lambda) \tilde v' =     \lambda \log v+(1-\lambda) \log v' \le \log(\lambda  v+(1-\lambda) v')=\log w=\tilde w,
    $$
so     $\lambda \tilde v+(1-\lambda) \tilde v' \in \dc(\tilde U)$.  Because $\tilde u \le \tilde v$ and $\tilde u' \le \tilde v'$, it follows that $\lambda \tilde u+(1-\lambda) \tilde u' \in \dc(\tilde U)$, as desired.
\end{proof}
Now we proceed to prove the theorem.\\

 \noindent \emph{The ``if'' direction:} Suppose that $u \in U$ is an SNBS over $U$ for some bargaining units $\mathcal I$ and  bargaining powers $\Psi$ (satisfying the requirement). Then,  $u \in V^T$ where $V^T=U$ and $V^t := \arg\max_{v \in V^{t-1}} \prod_{i\in I^t}  v_i^{\psi_i^t}$ for each $t \ge 1.$ This implies that $V^t := \arg\max_{v \in V^{t-1}} \sum_{i\in I^t}  \psi_i^t \log v_i.$
    Setting $\tilde u :=\log u$ and noting that $\psi$ is a nonnegative and eventually positive sequence, $\tilde u$ is a SUWM solution of $\dc(\tilde U)$ with respect to $\psi$. Therefore, by \Cref{theorem:characterize-PO}, $\tilde u$ is Pareto optimal in $\dc(\tilde U).$ Then, by \Cref{simple-observation}, $u$ is Pareto optimal with respect to $U$, as desired.\\
    
\noindent \emph{The ``only if'' direction:} Suppose that $u \in U$ is Pareto optimal with respect to $U$. Then, by \Cref{simple-observation}, $\tilde u := \log u$ is Pareto optimal with respect to $\dc(\tilde U).$ Therefore, by \Cref{theorem:characterize-PO}, there exist a sequence $\phi:=(\phi^t)_t$ of nonnegative and eventually positive welfare weight vectors such that $\tilde u \in \tilde U^T$ where $\tilde U^0=\dc(\tilde U)$ and
    $\tilde U^t := \arg\max_{\tilde u' \in \tilde U^{t-1}} \sum_{i\in I^t}  \phi_i^t \tilde u'_i$ for each $t\ge 1$. Then, for $u=e^{\tilde u}$, we have $u \in U^T$, where 
        $U^t := \arg\max_{u' \in U^{t-1}} \prod_{i\in I^t}   (u'_i)^{\phi_i^t}=\arg\max_{u' \in V^{t-1}} \prod_{i\in I^t}   (u'_i)^{\psi_i^t}$ for each $t \ge 1$, where $V^t:=\{e^{\tilde v}| \tilde v \in U^t\}$ and $\psi^t_i:=\frac{\phi^t_i}{\sum_{j \in I^t}\phi^t_j}$, so $u$ is an SNBS, as desired (note that $\psi$ satisfies the condition required of bargaining powers for SNBS).
\end{proof}

\subsection{Proof of \Cref{SNBS-axiomatization}} 
\label{s:snbs-axiomatization}
 To prove  (ii) implies (i), we only check that $R^*$ satisfies \textsf{Log Invariance} since the other axioms are rather straightforward to check. To do so, suppose that $u R^* v$  so that for some $t \le T$, $\prod_{i\in I^s}  u_i^{\psi_i^s} = (>) \prod_{i\in I^s}  v_i^{\psi_i^s}$ for   $s < (=) t +1$, which implies \begin{align}
     \label{log-comp}      \sum_{i \in I^s} \psi_i^s \ln u_i 
  =  (>)  \sum_{i \in I^s} \psi_i^s \ln  v_i \mbox{ for } s  < (=) t+1. 
 \end{align} Consider now any $u', v'$ such that    for some $a \in \mathbb{R}^n$ and $b \in \mathbb{R}_{++}$,  $\ln u'_i =a_i + b \ln u_i $ and $\ln v'_i  = a_i + b \ln v_i$ for all $i\in I$.   By \eqref{log-comp}, we have
 \begin{align*}
      \sum_{i \in I^s} \psi_i^s \ln u_i'   &  =   \sum_{i \in I^s}  \psi_i^s a_i  + b \left( \sum_{i \in I^s}  \psi_i^s \ln u_i  \right)  \\ & = (>)  \sum_{i \in I^s}  \psi_i^s a_i  + b \left( \sum_{i \in I^s}  \psi_i^s \ln v_i  \right) =  \sum_{i \in I^s} \psi_i^s \ln  v_i' \mbox{ for } s < (=)t+1,
 \end{align*}   which implies $\prod_{i\in I^s}  (u_i')^{\psi_i^s}  = (>) \prod_{i\in I^s}  (v_i')^{\psi_i^s}$ for $s < (=) t +1$ or $u' R^* v'$ as desired. 

 We now  prove that (i) implies (ii). Given any  $u \in \mathbb{R}^n$, let $e^u$ denote a vector $ (e^{u_i})_{i\in I}$ and $\ln u $ denote a vector $(\ln u_i)_{i\in I}$ for simplicity. Consider any welfare ordering $R^*$ on $\mathbb{R}_{++}^n$ that satisfies the three axioms. Let us define another ordering $\tilde R^*$ on $\mathbb{R}^n$ as follows:  for any $u, v \in \mathbb{R}^n$, $u \tilde R^* v$ if  $e^u R^* e^v$. 
 It is straightforward to check that $\tilde R^*$ satisfies \textsf{Pareto Principal}, \textsf{Invariace}, and \textsf{Weak Continuity}. 
 In particular, \textsf{Invariance} holds for the following reason.  Consider any $u,v$  such that $u \tilde R^* v$ or equivalently $\tilde u: = e^u R^* e^v =: \tilde v$. \textsf{Invariance} requires that for any $a \in \mathbb{R}^n$ and $b \in \mathbb{R}_{++} $,  $ u':= (a + b u) \tilde R^* ( a + bv) =: v'$ or equivalently $\tilde u' := e^{u'} R^* e^{v'}  =: \tilde v'$, which  follows from   $\tilde u R^* \tilde v$ and  \textsf{Log Invariance} since  $\ln \tilde u'  = a +b  \ln \tilde u $ and $\ln \tilde v' = a +b \ln  \tilde v$.  Since $\tilde R^*$ satisfies the \textsf{Pareto Principle}, \textsf{Invariance}, and \textsf{Weak Continuity}, \Cref{social-preference-theorem}  implies that there exists a nonnegative and eventually positive sequence of weight vectors $\Phi=(\phi^1, \phi^2, ..., \phi^T)$ such that for any $u,v \in \mathbb{R}^n$,  $u \tilde R^* v$ if and only if $u$ sequentially utilitarian welfare dominates $v$ according to $\Phi$. For each $t=1,\ldots, T$, let  $I^t = \supp{\psi^t}$ and  $\psi_i^t =\frac{\phi^t_i}{\sum_{i\in I^t} \phi^t_i}$ for all $i \in I^t$. Consider any $u,v$ with $u R^*v$. Then, $\tilde u:=\ln u \tilde R^* \ln v =: \tilde v$ so that $\tilde u$ sequentially utilitarian welfare dominates $\tilde v$ according to $\Psi = (\psi^1,\ldots,\psi^T)$: that is, for some $t \le T$, $\sum_{i \in I^s} \psi_i^s \tilde u_i  = (>)  \sum_{i \in I^s} \psi_i^s \tilde v_i$ for $s < (=) t+1$, which implies that $\prod_{i\in I^s}  u_i^{\psi_i^s} = (>) \prod_{i\in I^s}  v_i^{\psi_i^s}$ for   $s < (=) t +1$, meaning $u$ sequentially Nash welfare dominates $v$ acording to $\mathcal{I}$ and $\Psi$.  
 
 It is straightforward, and thus omitted, to prove that $u$ sequentially Nash welfare dominating $v$ according to $\mathcal{I}$ and $\Psi$ implies $u R^* v$.

\subsection{Proof of \Cref{prop:PLC-char}}
\label{sec:plc-cha-proof}

	\emph{The ``only if'' direction:}  By    \cref{theorem:characterize-PO},  for any Pareto optimal $u \in U \cap \mathbb{R}_{++}^n$,   there are nonnegative and eventually positive weights  $(\phi^1,\ldots,\phi^T)$ sequentially maximized by $u$. 
 Letting $U^t$ be defined as in \cref{eq:define-u-t}, we have $u\in U^t$ for all $t =1,\ldots,T$.      Consider  weights $(\psi^t)_{t=1}^T$ defined as  $\psi^1 =\phi^1$ and $\psi^t =  \frac{\langle  \psi^{t-1}, u \rangle }{\langle \phi^t,u \rangle } \phi^t$ for each $t\ge 2$. First, using  the fact that  $u \in \mathbb{R}_{++}^n$ and $ \phi^t  \in \mathbb{R}_+^n, \forall t$,  it is straightforward to see  that   $\langle  \psi^{t}, u \rangle  > 0$ and $\langle \phi^t,u \rangle  >0$ for every $t$.   Thus,    $\langle \psi^t, v  \rangle  \ge  \langle \psi^t, v' \rangle  $ if and only if $\langle \phi^t, v \rangle  \ge  \langle \phi^t, v' \rangle$ for all $v,v' \in U$. Note also  that  $\langle \psi^t , u \rangle  =     \frac{\langle \psi^{t-1}, u \rangle}{\langle  \phi^t, u \rangle}    \langle  \phi^t, u \rangle   =  \langle \psi^{t-1}, u\rangle$ for each $t \ge 2$. Thus, we have $W  (u) = \langle \psi^t, u \rangle$ for all $t =1, \ldots,T$. Also, for any  $v \in U^T$, we have  $\langle  \psi^t, u \rangle =\langle  \psi^t, v \rangle$ for all $t$, so  $W(u) = W(v)$.  For any $v \not \in U^T$, there is some $t$ such that $v  \not\in  U^t$  so  $\langle  \psi^t, v \rangle < \langle  \psi^t, u \rangle$, implying $W(v) < W(u)$.  Thus, $u$ maximizes $W$, implying that   $W$ achieves its maximum over $U$ via eventually positive weights.
	
\vskip 10pt
\noindent	\emph{The ``if'' direction:} Consider any $u \in U$  maximizing a  PLC function  $W$ that  achieves its maximum via eventually positive weights.  Suppose for contradiction that $u$ is not Pareto optimal.  Then, there is some $v > u$  so that  $\langle \psi^t, v  \rangle  \ge \langle \psi^t, u\rangle$ for all $t  = 1, \ldots, T$. As $u$ maximizes $W$, so does $v$.  Given this and the fact that   $W$ achieves    its maximum via eventually positive weights,  we must have   $ \langle \psi^T, v  \rangle  = W (v) = W (u) =\langle \psi^T, u  \rangle    $ or  $\langle  \psi^T , v -u  \rangle  =0 $, which is  a contradiction since $ \psi^T \gg 0$ and $v > u$.

\section{Pareto optimality ($U^P$) and positive utilitarianism ($U^{++}$)}\label{ss:chris-part}

This section aims to discover natural conditions for $U^P$ to coincide with $U^{++}$. The following lemma, which follows easily from the proof of \cref{theorem:characterize-PO}, is the key to our investigation. 

\begin{lemma}\label{cor:exposed-faces-easy}
If $u$ is a maximal element of $U$ that lies in the relative interior of an exposed face of $\text{dc}(U)$ then $u$ maximizes a positive weight vector over $U$. 
\end{lemma}
\begin{proof}
In the proof of the ``only if'' part of \cref{theorem:characterize-PO} in \cref{s:proof}, if $u$ is a maximal element of $U$ that lies in the relative interior of an exposed face of $\text{dc}(U)$, then $T =1$ in \cref{step3} and by \cref{step4} we know $\phi^1$ is positive. Hence, $\Phi =(\phi^1)$ and so by \cref{step5}, we conclude that $u$ maximizes the positive weight vector $\phi^1$ over $U$.
\end{proof}

To characterize when $U^P = U^{++}$, we need to introduce a few notions and establish their properties. 
First, the \emph{normal cone of $U$ at a point} $u \in U$ is the set 
\begin{equation*}
N_U(u) = \left\{\phi \in \R^n \mid \langle \phi, u \rangle \ge \langle \phi, v \rangle \text{ for all } v \in U \right\}.
\end{equation*}
If $\phi \in N_U(u)$ then $u$ is a maximizer of the linear function $\langle \phi, u \rangle$ over the set $U$. 
Then, \emph{the normal cone of a face} $F \subset U$, denoted $N_U(F)$, as the normal cone of each of its relative interior points. Next, the \emph{relative boundary of} $F$ is  defined as $F \setminus \mathrm{ri}(F)$.

The next two lemmas give us some properties of these notions.

\begin{lemma}\label{lemma:normal-cone-on-faces}
Let $F$ be a face of a convex set $U$. Then, every point in the relative interior of $F$ has the same normal cone. 
\end{lemma}
\begin{proof}
Let $u, u'$ be distinct in the relative interior of $F$ and suppose $N_U(u)$ contains an element $\phi$ not in $N_U(u')$. This implies $\langle \phi, u \rangle > \langle \phi, u' \rangle$. Since $u$ is the relative interior, the point $v = u + \lambda (u - u')$ lies in $F$ for a sufficiently small positive $\lambda$. However, $\langle \phi, v \rangle = \langle \phi, u \rangle + \lambda \langle \phi, u - u' \rangle >\langle \phi, u \rangle$, violating the assumption that $\phi$ is in $N_U(u)$.
\end{proof}

\begin{lemma}\label{lemma:boundaries-contain-more}
Let $F$ be a face of a convex set $U$. Then, every  point $u$ in the relative boundary of $F$ has $N_U(u) \supset N_U(F)$.
\end{lemma} 
\begin{proof}
Let $u$ be in the relative boundary of $F$. Suppose there is a weight vector $\phi$ in $N_U(v)$ (where $v$ is any relative interior element of $F$) that is not in $N_U(u)$. That is, 
\begin{equation}\label{eq:condition-to-contradict}
\langle \phi, u \rangle \neq \langle \phi, v \rangle.
\end{equation}
By the definition of the relative interior, we can get an element of the relative interior of $F$ arbitrarily close to $u$, which yields a contradiction of the continuity of $\langle \phi, \cdot \rangle$ because of \cref{eq:condition-to-contradict}.
\end{proof}

We are now ready to  provide the condition that characterizes when $U^P = U^{++}$:  
\begin{proposition}\label{theorem:characterize-maximal}
Let $U$ be a closed convex set. Then $U^P = U^{++}$ if and only if every maximal element of $U$ belongs to some exposed maximal face of $\text{dc}(U)$.
\end{proposition} 
\begin{proof}
\noindent \emph{The ``if'' direction.} Observe that $U^{++} \subset U^P$ is immediate from Proposition 3.23 in \cite{bewley2009general}. It remains to show that $U^P \subset U^{++}$. Let $u \in U^P$. If $u$ lies in the relative interior of an exposed face of $\text{dc}(U)$, then $u \in U^{++}$ from \cref{cor:exposed-faces-easy}. The remaining case is where $u$ lies on the relative boundary of a maximal exposed face $F$ of $\text{dc}(U)$. Since $F$ is a maximal exposed face, then an element $v$ in its relative interior maximizes a positive weight vector $\phi$, again by \cref{cor:exposed-faces-easy}. By \cref{lemma:normal-cone-on-faces}, this implies that the normal cone $N_U(F)$ of face $F$ contains $\phi$ and so, by \cref{lemma:boundaries-contain-more}, the normal cone $N_U(u)$ of the point $u$ contains $\phi$. In other words, $u$ maximizes the positive weight vector $\phi$. This completes the proof. 
\vskip 10pt
\noindent \emph{The ``only if'' direction.} Let $u$ be a maximal element of $U$. By the equivalence of $U^P$ and $U^{++}$, $u$ maximizes a positive weight vector $\phi$. Let $F = \arg\max_{v \in U} \langle \phi, v \rangle$. We claim that $F$ is a maximal exposed face of $\text{dc}(U)$, which contains $u$. The fact that $F$ is maximal in $\text{dc}(U)$ follows since Proposition 3.23 in \cite{bewley2009general} (along with \cref{lemma:maximal-face}) implies $F$ is maximal in $U$ and thus maximal in $\text{dc}(U)$ by \cref{lemma:downward-closure-suffices}. Suppose to the contrary that $F$ is not exposed in $\text{dc}(U)$. Then, there must exist an element $u' \in \text{dc}(U) \setminus U$ that maximizes $\phi$ but is not in $F$. However, since $u' \in \text{dc}(U) \setminus U$, there must exist $u'' \in U$ such that $u' \le u''$ and $u'_i < u''_i$ for some index $i$. But this implies that  $\langle \phi, u \rangle \ge \langle \phi, u'' \rangle > \langle \phi, u' \rangle$, where the weak inequality holds by the definition of $F$ and the strict inequality holds since $\phi$ is positive. This yields a contradiction and so we conclude that $F$ is an exposed face of $\text{dc}(U)$.  
\end{proof} 

We now discuss a few of the nuances in the statement of \cref{theorem:characterize-maximal}. First, the condition cannot be weakened so that every maximal element of $U$ simply lies in a (potentially nonmaximal) exposed face of $\text{dc}(U)$. Consider our canonical example in \cref{fig:illustration-of-sequential-converse}. The point $u$ lies on an exposed face of $\text{dc}(U)$, but this face is not a maximal face of $\text{dc}(U)$.

\cref{fig:illustration-of-sequential-converse} also demonstrates that it is not sufficient for a point  to lie on a maximal exposed face of $U$ (as opposed to $\text{dc}(U)$) to guarantee it maximizes a positive weight vector. Consider the point $u''$, which is a maximal exposed extreme point of $U$ but does not maximize any positive weight vector over $U$. However, $u''$ does not lie on a maximal exposed face of $\text{dc}(U)$ and so does not contradict the theorem.

 Given the above nuance, a simpler sufficient condition may be useful. Consider the  setting where all maximal faces of $\text{dc}(U)$ are exposed.  

\begin{corollary}\label{cor:maximals-and-exposed-faces}
If $U$ is a closed convex set such that all maximal faces of $\text{dc}(U)$ are exposed, then $U^P = U^{++}$.%
%
%
\footnote{This cannot be derived easily from \citet{arrow1953abb}. To see this, recall that they establish $U^{++} \subset U^P \subset \text{cl}(U^{++})$. This implies that if $U^{++}$ is closed then $U^P = U^{++}$. However, in the ``tilted cone'' in \cref{fig:tilted-cone}, $U^{++}$ is not closed since the point $K$ does not lie in $U^{++}$ but is the limit point of elements in $U^{++}$ (indicated by the line in the figure). However, it is straightforward to check that $U^P$ and $U^{++}$ coincide. One can also check that all maximal faces of $\text{dc}(U)$ for $U$ in \cref{fig:tilted-cone} are exposed, the condition of \cref{cor:maximals-and-exposed-faces}.}
%
%
 
\end{corollary}
\begin{proof}
Note that every maximal element of $U$ lies in a maximal face of $\text{dc}(U)$ by \cref{lemma:downward-closure-suffices}. This and the hypothesis imply that every maximal element of $U$ belongs to some exposed maximal face of $\text{dc}(U)$. Applying \cref{theorem:characterize-maximal}, we obtain the desired conclusion.
\end{proof}

\begin{figure}
\centering
\includegraphics[scale=.9]{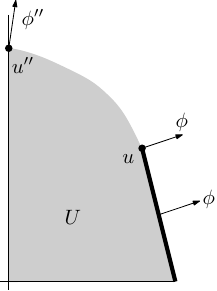}
\caption{The maximal extreme point $u$ is not exposed while $U^P = U^{++}$.}
\label{fig:counter-example-to-characterization}
\end{figure}

However, the converse of \cref{cor:maximals-and-exposed-faces} is false, as  illustrated by the example in \cref{fig:counter-example-to-characterization}.   
One sufficient condition for the hypothesis of   \cref{cor:maximals-and-exposed-faces} to hold is that $U$ is a polyhedron. In that case, all faces of $U$ are exposed  (Theorem 13.21 of \cite{soltan2015lectures}); moreover, its downward closure of a polyhedron is also a polyhedron  (Theorem 13.20 of \cite{soltan2015lectures}), so all of its faces are exposed.

Let $X$ be a polyhedral subset of $\R^m_+$ (possibly $\R^m_+$ itself). The utility function $ u_i:  X \to \R$ is \emph{piecewise-linear concave (PLC)}  if there exist finite index set $K_i$ and affine functions $u_{i,k}: \R^m_+ 
\to \R$ for each $k \in K_i$ such that $u_i (x) = \min_{k \in K_i} u_{i,k} (x)$ for all $x \in X$.\label{def:plc} The lemma uses some of the following facts.

\begin{lemma}\label{lemma:algebra-of-polyhedra}
The following properties on polyhedra hold:
\begin{enumerate}[label=(\roman*)]
\item Let $P_1, P_2, \ldots, P_n$ be a finite collection of polyhedra in $\R^m$. The Cartesian product $P_1 \times P_2 \times \cdots \times P_n$ is a polyhedron in $\R^{mn}$. 
\item Let $\pi: \R^d \to \R^m$ be an affine map and let $P$ be a polyhedron in $\R^d$. Then $\pi(P)$ is a polyhedron. 
\item All faces of a polyhedron are exposed. 
\item The downward closure of a polyhedron is also a polyhedron.  
\end{enumerate}
\end{lemma}
\begin{proof}
(i) Consider two polyhedra in $\R^m$, $P_1$ and $P_2$.  Letting $Q_1: =P_1 \times \R^m$ and $Q_2 := \R^m \times P_2$, each $Q_k$ is a polyhedron in $\R^{2m}$, so $P_1 \times P_2 = \cap_{k=1,2} Q_k $ is a polyhedron in $\R^{2m}$. The result follows from  applying this argument repeatedly.  (ii) This is Theorem 13.21 in \cite{soltan2015lectures}. (iii) This is Corollary 13.12 in \cite{soltan2015lectures}. (iv) This follows by nothing that since $\dc(P) = P + \R^n_-$ where $\R^n_-$ is the nonpositive orthant of $\R^n$ and by applying Theorem 13.20 of \cite{soltan2015lectures}. 
\end{proof}

\begin{lemma}\label{prop:piecewise-linear-concave-gives-dc-exposed}
If each agent has a PLC utility function defined on a  polyhedron $X$ and $U$  is  defined according to \cref{eq:X-U},  then $\mathrm{dc}(U)$ is a polyhedron. 
\end{lemma}
\begin{proof}
For each $k \in K_i$, let $X_{i,k} = \{   x \in X \mid  u_{i,k} (x)  \le u_{i, k'} (x),\forall k' \in K_i \}$. Since $X$ is a polyhedron and all functions $(u_{i,k})_{k\in K_i}$ are affine,   $X_{i,k}$ is an intersection of finitely many polyhedra  and thus  a polyhedron. 

Now let $\mathcal{K}  = \{ \mathbf k =  (k_i)_{i \in I} \mid k_i \in K_i \text{ for all } i \} $. For each $\mathbf k \in \mathcal{K}$,   let $X_{\mathbf k} = \cap_{i\in I} X_{i, k_i}$ and observe that  $X_{\mathbf k}$ is a polyhedron. Also, all functions $u_1 (\cdot),\ldots,u_I (\cdot)$ are affine on $X_{\mathbf k}$ since for each $i \in I$, $u_i (x) = u_{i, k_i} (x),\forall x \in X_{\mathbf k}$.   Then, by \cref{lemma:algebra-of-polyhedra}(ii), the set $U_{\mathbf k} = \{ (u_i (x))_{i\in I} \mid x \in X_{\mathbf k}  \}$ is a polyhedron.  Observe that  $U  = \{  (u_i (x))_{i\in I} \mid x \in X  \}  = \cup_{\mathbf k \in \mathcal{K}} U_{\mathbf k}$.  While we do not know whether the set $U$, which is a union of polyhedra,  is a polyhedron, Theorem 13.19 of \cite{soltan2015lectures} shows that $\overline{U}:= \mathrm{cl} (\mathrm{conv} \cup_{\mathbf k \in \mathcal{K}} U_{\mathbf k} )  $  is a polyhedron, where $\mathrm{cl}$ and $\mathrm{conv}$ denote the closure and convex hull, respectively.  

Next, we show that  $\mathrm{dc} (U) =\mathrm{dc} (\overline{U}) $. 
By definition of $\overline{U}$,  $\mathrm{dc} (U) \subset \mathrm{dc} (\overline{U}) $ is clear. To show $  \mathrm{dc} (\overline{U}) \subset \mathrm{dc} (U) $, consider any $\tilde u \in \mathrm{conv} \cup_{\mathbf k \in \mathcal{K}} U_{\mathbf k}  $ so that   $\tilde u = \sum_{\mathbf k \in \mathcal{K}} \lambda_{\mathbf k} \tilde u_{\mathbf k}$ for some weight $(\lambda_{\mathbf k})_{\mathbf k \in \mathcal{K}}$ and $\tilde u_{\mathbf k} \in \cup_{\mathbf k' \in \mathcal{K}} U_{\mathbf k'}$.  Also, for each $\tilde u_{\mathbf k}$, we can find $\tilde x_{\mathbf k} \in X_{\mathbf k}$ such that $(u_i (\tilde x_{\mathbf k}))_{i\in I} = \tilde u_{\mathbf k}$. Letting $x = \sum_{\mathbf k\in \mathcal{K}} \lambda_{\mathbf k} \tilde x_{\mathbf k}$, observe that $x\in X$ by the convexity of $X$ and that for all $i \in I$, $u_i (x) \ge \sum_{\mathbf k\in \mathcal{K}} \lambda_{\mathbf k} u_i (\tilde x_{\mathbf k}) = \tilde u_i$  by the concavity of $u_i (\cdot)$, which means that $\tilde u \in \mathrm{dc} (U)$. Thus, $\mathrm{conv} \cup_{\mathbf k \in \mathcal{K}} U_{\mathbf k}  \subset \mathrm{dc}(U)$, implying that $  \mathrm{cl} (\mathrm{conv} \cup_{\mathbf k \in \mathcal{K}} U_{\mathbf k} )  \subset  \mathrm{dc}(U)$ since $\mathrm{dc} (U)$ is closed, from which    $\mathrm{dc} (\overline{U})  \subset \mathrm{dc} (U)$ follows, as desired.   

Lastly, observe that $\mathrm{dc} (\overline{U})  = \overline{U} + \mathbb
R_-^n$ and that both $\overline{U}$ and $\mathbb{R}_-^n$ are polyhedra, which implies (by \cref{lemma:algebra-of-polyhedra}(iv)) that  $\mathrm{dc} (\overline{U})  = \mathrm{dc}(U)$ is a polyhedron. 
\end{proof}

The following is obtained immediately from \cref{cor:maximals-and-exposed-faces,prop:piecewise-linear-concave-gives-dc-exposed}, and the fact that all faces of polyhedra are exposed. It is a clean economic setting where $U^P$ and $U^{++}$ coincide. 

\begin{proposition}\label{theorem:plc-positive-utility}
If each agent has a PLC utility function defined on a polyhedron   $X$ and $U$ is defined according to \cref{eq:X-U}, then $U^P = U^{++}$.%
%
%
\footnote{\label{footnote:polyhedron}It is worth noting that the ABB theorem provides an alternative proof of this result. Recall that it suffices to argue $U^{++}$ is closed in order to conclude $U^P = U^{++}$. Clearly, the elements of $U^{++}$ come in faces, and a polyhedron has finitely many faces. Since the faces of a polyhedron are closed, and a finite union of closed sets is closed, this implies that $U^{++}$ is closed.}
%
%
\end{proposition}

\section{Second welfare theorem with piecewise-linear concave utility functions} \label{s:second-welfare-thm}  

In the paper, we showed that the notions of exposed faces and normal vectors play crucial roles for our characterization of a Pareto optimal utility profile as a welfare-maximizing point. Recall that the  normal vector also plays an important role in the second theorem of welfare economics in identifying  a price vector  that supports a Pareto optimal allocation as a competitive equilibrium outcome. Unlike in our characterization, the idea of a normal vector   in the second welfare theorem applies to the space of goods, not the space of utility profiles. However, the fact that the two spaces are closely connected hints at the possibility of establishing the second welfare theorem using the machinery we have developed so far.   We do so in the current section under a set of assumptions on the agent preferences and endowments that generalize the existing welfare theorem in a certain direction.        

To begin,  consider an exchange economy with $m$ types of goods with some integer $m>0$. For each $k \in \{1,\dots,m\}$, let $\bar e^k>0$ be the total supply of type-$k$ goods in the environment. Let $\bar e$ denote the vector $(\bar e^k)_{k =1}^m$. Each alternative $x=(x_i)_{i \in I}$, $x_i=(x^k_i)_{k =1}^m \in \mathbb R^m_+$, specifies consumption bundle $x_i$ for each $i \in I$. A profile of consumption bundles $x$ is said to be feasible if and only if $\sum_{i \in I} x_i \le \bar e$. In this context, the choice set $X$ is defined as the set of all feasible profiles of consumption bundles.  Each individual $i \in I$ is endowed with a utility function $u_i:\mathbb R^m_+ \to \mathbb R$. Suppose that each agent $i$ is endowed with a vector of goods $e_i \in  \R^m_+ \backslash \{0\}$  and let $\bar e = \sum_{i \in I} e_i $. A vector $p \in \mathbb R^m$ is referred to as a price profile.  A pair $(p, x)$ of a price profile $p$ and a profile $x=(x_i)_{i \in I}$ of consumption bundles is a \emph{Walrasian equilibrium} if 
\begin{enumerate}
\item $\sum_{i \in I} x_i = \bar e,$ and
\item  $x_i \in \arg\max_{y_i \in B_i(p)} u_i(y_i)$ for  each $i \in I$, where $B_i(p) := \{y_i \in \mathbb R^m_+ \mid \langle p, y_i \rangle \le \langle p, e_i \rangle \}$ is the budget set of $i$.
\end{enumerate}

We consider a case where utility functions of all players are piecewise-linear concave (PLC), as defined in  \cref{def:plc}. PLC utility functions may appear somewhat restrictive, but  any concave function can be approximated arbitrarily closely by a PLC utility function \citep{broiva75}. Meanwhile, we make a weaker assumption  in another dimension---preference monotonicity. The existing second   welfare theorem assumes agents' utility functions to be strictly monotonic. We invoke a weaker form of monotonicity. Say that an allocation  $(x_i)_{i \in I}$ is \emph{strictly feasible for good $k$} if it is feasible and satisfies $\sum_{i\in I} x_{i}^k  <  \bar e^k$.  We   assume that  the agent preferences are  \emph{monotonic under limited resources} in the following sense:  for any allocation $ (x_i)_{i\in I}$ that is strictly feasible for good $k$, there exist an agent $j$ and $\tilde x_j \in \R^m_+$ such that   $u_j (\tilde x_j) >  u_j (x_j)$   while $\tilde x_{j}^{k'} =x_{j}^{k'}, \forall k' \ne k$, $\tilde x_{j}^{k} > x_{j}^{k}$, and $\tilde x_{j}^{k} + \sum_{i\ne j} x_{j}^{k} \le \bar e^k$.  That is,   given any allocation that does not exhaust the endowment of good  $k$, there exists an agent who gets better off by consuming more of that good within its endowment. This condition is fairly weak. For instance, it allows for agents to consider a certain good indifferently, or even as bads (rather than goods), as long as there is at least one agent who likes to consume that good.  We are now ready to prove the second welfare theorem under the above assumptions. 

\begin{proposition}\label{theorem:second-welfare-piecewise-linear}
Consider the exchange economy described above. 
 If $(u_i(e_i))_{i \in I}$ is  Pareto optimal,  then there exists a positive price vector $p \gg 0$  such that $(p, (e_i)_{i \in I})$ is a Walrasian equilibrium. 
\end{proposition}

\begin{proof}[Proof of \cref{theorem:second-welfare-piecewise-linear}]
Let  $A_i := \{x \in \R_+^m \mid u_i (x) \ge u_i (e_i) \}$ for each agent $i$. Observe that each $A_i$ is a polyhedron since it is an intersection of two polyhedra,  $\{x \in \R^m \mid x \ge 0 \}$ and $\{x \in \R^m \mid u_i (x) \ge u_i (e_i) \} = \cap_{k \in K_i}    \{x \in \R^m \mid u_{i,k} (x) \ge u_i (e_i) \}$.

Consider the set $A = \left\{x \in \R^m_+ \mid \exists x_1 \in A_1, x_2 \in A_2, \dots, x_n \in A_n \text{ s.t. } x= \sum_{i \in I} x_i  \right\}$. Observe that $A$ is the image of the set $A_1 \times A_2 \times \cdots \times A_n$ under the affine mapping $\pi$ that maps $(x_i)_{i \in I}$ to $\sum_{i \in I} x_i$. By \cref{lemma:algebra-of-polyhedra}(i) and (ii), $A$ itself is a polyhedron. 

Next, we argue that $\bar e$ is a minimal element of the set $A$. Suppose for contradiction that there exists an element $x \in A$ where $x < \bar e$ where $x^k < \bar e^k$ for some good $k$.   Since $x \in A$, there exists an allocation $(y_i)_{i\in I} $ where $y_i \in A_i$ such that $x = \sum_{i \in I} y_i$.  Since this allocation is strictly feasible for the good $k$, the monotone preference under limited resources implies  that there are some agent $j$ and $\tilde y_j \in \R^m_+$ such that $u_j (y_j) < u_j (\tilde y_j)$ while  $\tilde y_{j}^{k'} =y_{j}^{k'}, \forall k' \ne k$, $\tilde y_{j}^{k} > y_{j}^{k}$, and $\tilde y_{j}^{k} + \sum_{i\ne j} y_{i}^{k} \le \bar e^{k} $.  Now consider an alternative allocation  $(z_i)_{i \in I}$, which is identical to $(y_i)_{i \in I}$ except that $z_j  = \tilde y_j$. Note that this allocation is feasible under the endowment $\bar e$ and that $u_j (z_j) > u_j (y_j) \ge u_j (e_j)$ while $u_i (z_i) = u_i (y_i) \ge u_i (e_i), \forall i \ne j$, which contradicts the Pareto optimality of  $(e_i)_{i \in I}$.  

That $\bar e$ is a minimal element of $A$ implies  that $-\bar e$ is a maximal element of $-A$. By \cref{lemma:downward-closure-suffices}, this implies that $-\bar e$ is a maximal element of $\text{dc}(-A)$. Moreover, by \cref{lemma:algebra-of-polyhedra}(iv) $\text{dc}(-A)$ is a polyhedron and so by \cref{lemma:algebra-of-polyhedra}(iii) all of its faces are exposed. Thus, by \cref{cor:exposed-faces-easy}, there exists a supporting hyperplane of $-A$ through the point $-\bar e$ with a positive normal $\phi$.  The same normal $p := \phi$ can define a supporting hyperplane to $A$ through the point $\bar e$; that is, 
\begin{equation*}
\langle p, y \rangle \ge  \langle p , \bar e \rangle, \forall y \in A,
\end{equation*}
where $p$ is a   positive vector of prices. 

It remains to show that the positive price vector $p$ just constructed supports the allocation $(e_i)_{i \in I}$ as a Walrasian equilibrium.   For this, it suffices to show that each $e_i$ maximizes $u_i (\cdot)$ under the prices $p$ and the  budget $\langle p, e_i \rangle$. To do so, we take any $x_i$ with  $u_i(x_i) > u_i(e_i)$ and show that agent $i$ cannot afford $x_i$. 

By continuity of $u_i$, the inequality $u_i(x_i) > u_i(e_i)$ implies that for some $\lambda < 1$ but sufficiently close to $1$, we have $u_i(\lambda x_i) > u_i(e_i)$, so by definition we have $\lambda x_i \in A_i$. This implies that $\lambda x_i + \sum_{j \neq i} e_j \in A$. Since $\langle p, \lambda x_i + \sum_{j\neq i} e_j \rangle  \ge \langle p, \sum_{i \in I} e_i \rangle$, we must also have $\langle p, \lambda x_i \rangle \ge \langle p, e_i \rangle$. Dividing through by $\lambda$ gives $\langle p, x_i \rangle \ge \langle  \frac{1}{\lambda} p, e_i \rangle > \langle p, e_i \rangle$ where the strict inequality holds since $e_i$ is nonnegative and nonzero while $p$ is   positive.\end{proof}

In addition to the weakening of preference monotonicity, we also dispense with the typical assumption required by the existing second welfare theorem that every consumer has a positive endowment for every type of good (i.e., $e_{i}\gg 0,\forall i \in I$).  The positive endowment assumption can be quite restrictive, excluding many realistic situations. Relaxing the same assumption was an important motivation behind Arrow's generalization of the first welfare theorem.\footnote{``While listening to a talk about housing by Franko (sic) Modigliani, Arrow realized that most people consume nothing of most goods (for example, living in just one particular kind of house), and thus that the prevailing efficiency proofs assumed away all the realistic cases,'' according to Geanakoplos in  \url{https://www.econometricsociety.org/sites/default/files/inmemoriam/arrow_geanakoplos.pdf}.
}  
%
%
At the same time, the theorem assumes PLC utility functions. This assumption guarantees  that the ``upper contour set'' of the target allocation---or the set of goods weakly preferred to  $(e_i)_{i \in I}$---is a polyhedron.  Meanwhile, preference monotonicity and Pareto-optimality of $(u_i(e_i))_{i \in I}$ ensure that the vector  $\bar e$ is a (minimal) face  of this set.    Invoking \cref{theorem:plc-positive-utility},  $\bar e$  is then exposed by a positive normal (or price vector) that supports  $(e_i)_{i\in I}$ as a competitive equilibrium allocation.

\end{document}